\documentclass[preprint,12pt]{elsarticle}

\usepackage{amssymb}
\usepackage{enumerate}
\usepackage{color}
\newcommand{\cC}{{\mathcal C}}
\newcommand{\cE}{{\mathcal E}}
\newcommand{\cF}{{\mathcal F}}

\newcommand{\cL}{{\mathcal L}}
\newcommand{\cN}{{\mathcal N}}
 
\newcommand{\cP}{{\mathcal P}}

\newcommand{\cS}{{\mathcal S}}
\newcommand{\cT}{{\mathcal T}} 
\newcommand{\cU}{{\mathcal U}} 

\newcommand{\ba}{\backslash}

\newtheorem{theorem}{Theorem}[section]
\newtheorem{lemma}[theorem]{Lemma}

\newtheorem{corollary}[theorem]{Corollary}

\newtheorem{observation}[theorem]{Observation}
\newproof{proof}{Proof}

\journal{Advances in Applied Mathematics}

\begin{document}

\begin{frontmatter}



\title{Attaching leaves and picking cherries to characterise the hybridisation number for a set of phylogenies}


\author[label1]{Simone Linz}
\ead{s.linz@auckland.ac.nz}

\author[label2]{Charles Semple}
\ead{charles.semple@canterbury.ac.nz}

\address[label1]{Department of Computer Science, University of Auckland, New Zealand}
\address[label2]{School of Mathematics and Statistics, University of Canterbury, New Zealand}

\begin{abstract}
Throughout the last decade, we have seen much progress towards characterising and computing the minimum hybridisation number for a set $\cP$ of rooted phylogenetic trees. Roughly speaking, this minimum quantifies the number of hybridisation events needed to  explain a set of phylogenetic trees by simultaneously embedding them into a phylogenetic network. From a mathematical viewpoint, the notion of agreement forests is the underpinning concept for almost all results that are related to calculating the minimum hybridisation number for when $|\cP|=2$. However, despite various attempts, characterising this number in terms of agreement forests for $|\cP|>2$ remains elusive. In this paper, we characterise the minimum hybridisation number for when $\cP$ is of arbitrary size and consists of not necessarily binary trees. Building on our previous work on cherry-picking sequences, we first establish a new characterisation to compute the minimum hybridisation number in the space of tree-child networks. Subsequently, we show how this characterisation extends to the space of all rooted phylogenetic networks. Moreover, we establish a particular hardness result that gives new insight into some of the limitations of agreement forests.
\end{abstract}


\end{frontmatter}

%\linenumbers

\section{Introduction}

In our quest for faithfully describing evolutionary histories, we are currently witnessing a shift from the representation of ancestral histories by phylogenetic (evolutionary) trees towards phylogenetic networks. The latter not only represent speciation events but also non-tree like events such as hybridisation and horizontal gene transfer that have played an important role throughout the evolution of certain groups of organisms as for example in plants and fish~\cite{drezen16,mallet16,marcussen14,soucy15}. 

In this paper, we focus on a problem that is related to the reconstruction of phylogenetic networks. Called {\sc Minimum Hybridisation} and formally stated at the end of this section, this problem was first introduced by Baroni et al.~\cite{baroni05}. While {\sc Minimum Hybridisation} was historically motivated by attempting to quantify hybridisation events, it is now more broadly regarded as a tool to quantify all non-tree like events to which we collectively refer to  as reticulation events.  Pictorially speaking, {\sc Minimum Hybridisation} aims at the reconstruction of a phylogenetic network that simultaneously embeds a given set of phylogenetic trees while minimising the number of reticulation events that are represented by vertices in the network whose in-degree is at least two. More formally, the problem is based on the following underlying question. Given a collection $\cP$ of rooted phylogenetic trees on the same set of taxa that have correctly been reconstructed for different parts of the species' genomes, what is the smallest number of reticulation events that is needed to explain $\cP$?~Over the last ten years, we have seen significant progress in characterising and computing this minimum number for when $|\cP|=2$ (e.g. see~\cite{albrecht12,bordewich07,bordewich07b,chen13,kelk12,wu10}). However, except for some heuristic approaches~\cite{chen12,wu10b}, less is known for when $|\cP|\geq 3$. This is due to the fact that the notion of agreement forests, which underlies almost all results that are related to {\sc Minimum Hybridisation}, appears to be ungeneralisable to more than two trees.  

Previously, together with Humphries, we introduced cherry-picking sequences and characterised a restricted version of {\sc Minimum Hybridisation} for $\cP$ being binary and of arbitrary size~\cite{humphries13}. Instead of minimising the number of reticulation events needed to explain $\cP$ over the space of all rooted phylogenetic networks, this restricted version only considers binary temporal tree-child networks. Such networks are the binary intersection of the classes of temporal networks and tree-child networks introduced by Moret et al.~\cite{moret04} and Cardona et al.~\cite{cardona12}, respectively.
%Temporal networks satisfy several time constraints and, 
Disadvantageously, this restriction is so strong that not even if $|\cP|=2$ are we guaranteed to have a solution, i.e. there may be no such network explaining $\cP$~\cite[Figure 2]{humphries13a}.
%the time constraints are so strong that not every collection of trees can be guaranteed to have a solution, i.e. there may be no temporal \plum{tree-child} network that explains $\cP$~\cite{humphries13}. \marginpar{\blue{Is temporal sufficient here?}}

Here, we advance our work on cherry-picking sequences and establish two new characterisations to quantify the amount of reticulation events that are needed to explain a set of (not necessarily binary) phylogenetic trees.~The first characterisation solves the problem over the space of tree-child networks. Unlike temporal networks, we show that every collection $\cP$ of rooted phylogenetic trees  has a solution, i.e. the trees in $\cP$ can simultaneously be embedded into a tree-child network. Subsequently, we extend this characterisation to the space of all rooted phylogenetic networks and, hence, provide the first characterisation for {\sc Minimum Hybridisation} in its most general form. Both characterisations are based on computing a cherry-picking sequence for $\cP$, while the latter characterisation makes also use of an operation that attaches auxiliary leaves to the trees in $\cP$.

In addition to the two new characterisations, we return back to agreement forests and investigate why they seem to be of limited use to solve {\sc Minimum Hybridisation} for an arbitrary size set $\cP$ of rooted phylogenetic trees. Roughly speaking, given $\cP$, one can compute a particular type of agreement forest $\cF$ of smallest size and, if $|\cP|=2$, then each but one component in $\cF$ contributes exactly one to the minimum number of reticulation events that is needed to explain $\cP$. On the other hand, if $|\cP|>2$, the contribution of each component in $\cF$ to this minimum number is much less clear. Motivated by this drawback of agreement forests, we consider a set $\cP$ of rooted binary phylogenetic trees as well as the agreement forest $\cF$ {\it induced} (formally defined in Section~\ref{sec:scoring}) by a phylogenetic network that explains $\cP$ and minimises the number of reticulations events and ask whether or not, it is computationally hard to calculate the minimum number of reticulation events that is needed to explain $\cP$. We call the associated decision problem {\sc Scoring Optimum Forest}. This problem was first mentioned in~\cite{iersel16}, where the authors conjecture that {\sc Scoring  Optimum Forest} is NP-complete. Using the machinery of cherry-picking sequences, we show that {\sc Scoring  Optimum Forest} is NP-complete for when one considers the smaller space of tree-child networks.

The paper is organised as follows. The remainder of the introduction contains some definitions and preliminaries on phylogenetic networks. In Section~\ref{sec:main}, we state the two new characterisations in terms of cherry-picking sequences. The first optimises {\sc Minimum Hybridisation} within the space of tree-child networks and the second optimises {\sc Minimum Hybridisation} within the space of all phylogenetic networks. The second characterisation is an extension of the first by additionally allowing the attachment of auxiliary leaves. We then establish proofs for both characterisations in Section~\ref{pickingproofs} as well as a formal description of the analogous algorithm. In Section~\ref{sec:bound}, we  establish an upper bound on the number of auxiliary leaves that, given a collection of phylogenetic trees, are needed to characterise {\sc Minimum Hybridisation} over the space of all rooted phylogenetic networks. Lastly, in Section~\ref{sec:scoring}, we formally state the problem {\sc Scoring  Optimum Forest} and show that it is NP-complete. We finish the paper with some concluding remarks in Section~\ref{sec:conclu}.

Throughout the paper, $X$ denotes a non-empty finite set. A {\it phylogenetic network $\cN$ on $X$} is a rooted acyclic digraph with no parallel edges that satisfies the following properties:
\begin{enumerate}[(i)]
\item the (unique) root has out-degree two,
\item the set $X$ is the set of vertices of out-degree zero, each of which has in-degree one, and
\item all other vertices either have in-degree one and out-degree two, or in-degree at least two and out-degree one.
\end{enumerate}
For technical reasons, if $|X|=1$, we additionally allow $\cN$ to consist of the single vertex in $X$. The set $X$ is the {\it leaf set} of $\cN$ and the vertices in $X$ are called {\it leaves}. We sometimes denote the leaf set of $\cN$ by $\cL(\cN)$. For two vertices $u$ and $v$ in $\cN$, we say that $u$ is a {\it parent} of $v$ and $v$ is a {\it child} of $u$ if $(u,v)$ is an edge in $\cN$. Furthermore, the vertices of in-degree at most one and out-degree two are {\it tree vertices}, while the vertices of in-degree at least two and out-degree one are {\it reticulations}. An edge directed into a reticulation is called a {\it reticulation edge} while each non-reticulation edge is called a {\it tree edge}. We say that $\cN$ is {\it binary} if each reticulation has in-degree exactly two. Lastly, a directed path $P$ in $\cN$ ending at a leaf is a {\it tree path} if every intermediate vertex in $P$ is a tree vertex.

A phylogenetic network $\cN$ on $X$ is {\it tree-child} if each non-leaf vertex in $\cN$ is the parent of at least one tree vertex or leaf. An example of two tree-child networks $\cN$ and $\cN'$ is given at the bottom of Figure~\ref{fig:cps}. Note that the phylogenetic network obtained from $\cN$ by deleting the leaf labelled 4 and suppressing the resulting degree-two vertex $v$ results in a network that is not tree-child.

A {\it rooted  phylogenetic $X$-tree $\cT$} is a rooted tree with no degree-two vertices except possibly the root which has degree at least two, and with leaf set $X$. If $|X|=1$, then $\cT$ consists of the single vertex in $X$. As for phylogenetic networks, the set $X$ is called the {\it leaf set} of $\cT$ and is denoted by $\cL(\cT)$. In addition, $\cT$ is {\it binary} if $|X|=1$ or, apart from the root which has degree two, all interior vertices have degree three.  
Since we are only interested in {\it rooted} phylogenetic trees and {\it rooted} binary phylogenetic trees in this paper, we will refer to such trees simply as phylogenetic trees and binary phylogenetic trees, respectively. For a  phylogenetic $X$-tree $\cT$, we consider two types of subtrees. Let $X'$ be a subset of $X$. The {\it minimal  subtree} of $\cT$ that connects all the leaves in $X'$ is denoted by $\cT(X')$. Moreover, the {\it restriction of $\cT$ to $X'$}, denoted by $\cT|X'$, is the phylogenetic $X'$-tree obtained from $\cT(X')$ by suppressing all degree-two vertices apart from the root. Lastly, for two phylogenetic $X$-trees $\cT$ and $\cT'$, we say that $\cT'$ is a {\it refinement} of $\cT$ if $\cT$ can be obtained from $\cT'$ by contracting a possibly empty set of internal edges in $\cT'$. In addition, $\cT'$ is a {\it binary refinement} of $\cT$ if $\cT'$ is binary.

Let $\cT$ be a phylogenetic $X'$-tree. A phylogenetic network $\cN$ on $X$ with $X'\subseteq X$ {\it displays} $\cT$ if, up to suppressing vertices with in-degree one and out-degree one, there exists a binary refinement of $\cT$ that can be obtained from $\cN$ by deleting edges, leaves not in $X'$, and any resulting vertices of out-degree zero, in which case we call the resulting acyclic digraph an {\it embedding} of $\cT$ in $\cN$. If $\cP$ is a collection of phylogenetic $X$-trees, then $\cN$ {\it displays} $\cP$ if each tree in $\cP$ is displayed by $\cN$. For example, the two phylogenetic networks at the bottom of Figure~\ref{fig:cps} both display each of the four trees shown in the top part of the same figure.

\begin{figure}[t]
\center
\scalebox{1.2}{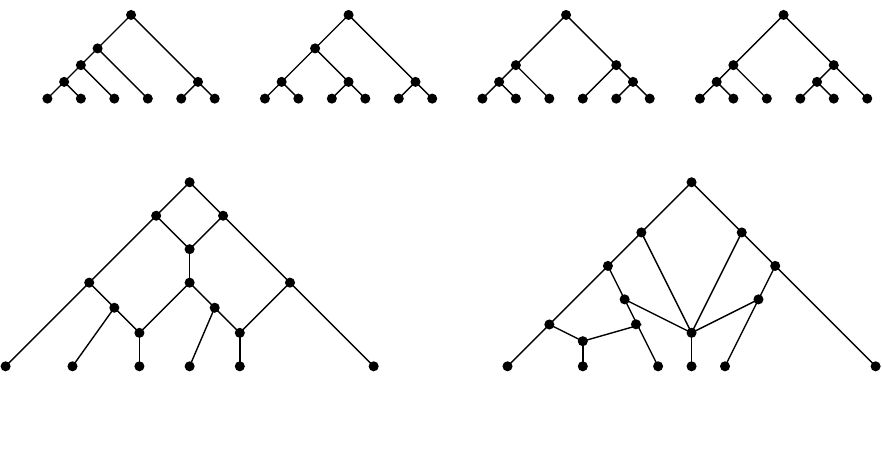}
\caption{Top: A set $\cP$ of four  phylogenetic $X$-trees with $X=\{1,2,\ldots,6\}$. Bottom: Two tree-child networks displaying $\cP$ with $h(\cN)=3$ and $h(\cN')=4$.}
\label{fig:cps}
\end{figure}

Let $\cN$ be a phylogenetic network with vertex set $V$ and root $\rho$. The {\it hybridisation number} of $\cN$, denoted $h(\cN)$, is the value
$$h(\cN)=\sum_{v\in V-\{\rho\}} \left(d^-(v)-1\right),$$
where $d^-(v)$ denotes the in-degree of $v$. For example, the phylogenetic networks $\cN$ and $\cN'$ that are shown in Figure~\ref{fig:cps} have hybridisation number 3 and 4, respectively. Observe that each tree vertex and each leaf contributes zero to this sum, but each reticulation $v$ contributes $d^-(v)-1$. Furthermore, for a set $\cP$ of phylogenetic $X$-trees, we denote by $h_{\rm tc}(\cP)$ and $h(\cP)$, respectively, the values
$$\min\{h(\cN):~\mbox{$\cN$ is a tree-child network on $X$ that displays $\cP$}\}$$
and
$$\min\{h(\cN):~\mbox{$\cN$ is a phylogenetic network on $X$ that displays $\cP$}\}.$$

\noindent {\bf Remark.} While the above definition of a phylogenetic network is restricted to networks whose tree vertices have out-degree exactly two, we note that the results in this paper also hold for networks with tree vertices whose out-degree is at least two. More particularly, if a set $\cP$ of phylogenetic $X$-trees is displayed by a phylogenetic network $\cN$ whose tree vertices have out-degree at least two, then, by ``refining'' such vertices, we can obtain a phylogenetic network $\cN'$ whose tree vertices have out-degree exactly two, displays $\cP$, and $h(\cN')=h(\cN)$. Thus no generality is lost with this restriction.\\

We next formally state the two decision problems that this paper is centred around.\\

\noindent{\sc Minimum Tree-Child Hybridisation} \\
{\bf Instance.} A set $\cP$ of phylogenetic $X$-trees and a positive integer $k$. \\
{\bf Question.} Does there exist a tree-child network $\cN$ on $X$ that displays $\cP$ such that $h(\cN)\leq k$?\\

\noindent{\sc Minimum Hybridisation} \\
{\bf Instance.} A set $\cP$ of phylogenetic $X$-trees and a positive integer $k$. \\
{\bf Question.} Does there exist a phylogenetic network $\cN$ on $X$ that displays $\cP$ such that $h(\cN)\leq k$?\\

We will see at the end of this section that, for any given set $\cP$ of phylogenetic $X$-trees,  \noindent{\sc Minimum Tree-Child Hybridisation} has a solution, i.e. there exists a tree-child network that displays $\cP$.

It was shown in~\cite{bordewich07} that {\sc Minimum Hybridisation} is NP-hard, even for when $\cP$ consists of two rooted binary phylogenetic $X$-trees. To see that {\sc Minimum Tree-Child Hybridisation} is also computationally hard, we again consider this restricted version of the problem and recall the following observation that was first mentioned in~\cite{humphries13} and can be derived by slightly modifying the proof of~\cite[Theorem 2]{baroni05}.

\begin{observation}\label{ob:tc}
Let $\cP=\{\cT,\cT'\}$ be a collection of two binary phylogenetic $X$-trees. If there exists a phylogenetic network $\cN$ that displays $\cP$ with $h(\cN)=k$, then there also exists a tree-child network $\cN'$ that displays $\cP$ with $h(\cN')\leq k$.
\end{observation}

\noindent The next theorem, whose straightforward proof is omitted, follows from Observation~\ref{ob:tc} and the fact that, given a tree-child network $\cN$ and a binary phylogenetic tree $\cT$, it can be checked in polynomial time whether or not $\cN$ displays $\cT$~\cite{iersel10,simpson}. 

\begin{theorem}\label{t:hardness}
The decision problem {\sc Minimum-Tree-Child Hybridisation} is {\rm NP}-complete.
\end{theorem}

%We end this section by showing that every collection of phylogenetic trees can be displayed by a tree-child network. For $n=2$, let $\cU_2$ be the unique binary phylogenetic tree on two leaves, $x_1$ and $x_2$ say, whose root  $\rho$ is a vertex at the end of a pendant edge adjoined to the original root. Now for a positive integer $n>2$, obtain $\cU_n$ from $\cU_{n-1}$ by adding an edge that joins a new  vertex $v$ and a new leaf $x_n$ and, for each tree edge $e$ in $\cU_{n-1}$, subdividing $e$ with a new vertex $u_e$ and adding the edge $(u_e,v)$. We call $\cU_n$ the {\it universal network} on $n$ leaves and note that $\cU_n$ is unique up to relabelling its leaves.

We end this section by showing that every collection of phylogenetic $X$-trees can be displayed by a tree-child network on $X$. For $n=2$, let $\cU_2$ be the unique binary phylogenetic tree on two leaves, $x_1$ and $x_2$ say. Now, for a positive integer $n > 2$, obtain $\cU_n$ from $\cU_{n-1}$ as follows. Viewing the root $\rho$ of $\cU_{n-1}$ as a vertex of in-degree zero and out-degree one adjoined to the original root, add an edge that joins a new  vertex $v$ and a new leaf $x_n$ and, for each tree edge $e$ in $\cU_{n-1}$, subdivide $e$ with a vertex $u_e$, and add the edge $(u_e, v)$. The resulting phylogenetic network without viewing the root as a vertex of in-degree zero and out-degree one is $\cU_n$.

\begin{theorem}\label{t:all}
Let \mbox{} $\cU_n$ be the universal network on $X=\{x_1,x_2,\ldots,x_n\}$ with $n\geq 2$. Then $\cU_n$ is tree-child and displays all binary phylogenetic $X$-trees.
\end{theorem}

\begin{proof}
By construction of $\cU_n$ from $\cU_{n-1}$ it is straightforward to check that, as $\cU_2$ is tree-child, $\cU_n$ is tree-child. To see that $\cU_n$ displays all binary phylogenetic $X$-trees, we use induction on $n$. Clearly, $\cU_2$ displays the unique binary phylogenetic tree on two leaves. For $n\geq 3$, assume that the universal network $\cU_{n-1}$ on $X'=\{x_1,x_2,\ldots,x_{n-1}\}$ displays all binary phylogenetic $X'$-trees. Observe that $\cU_{n-1}$ can be obtained from $\cU_n$ by deleting $x_n$, the parent of $x_n$ and all their incident edges, and suppressing all resulting vertices with in-degree one and out-degree one.
%Assume that each universal network $\cU_{n'}$ with $n'<n$ displays all binary phylogenetic trees on the same leaf set as $\cU_{n'}$. Let $p$ be the parent of $x_n$ in $\cU_n$. By construction, $p$ is a reticulation. Now, obtain $\cU_{n-1}$ from $\cU_n$ by deleting $p$ and all edges that are incident with $p$, and suppressing all resulting vertices with in-degree 1 and out-degree-1. It follows from the induction assumption that $\cU_{n-1}$ displays all binary phylogenetic trees with leaf set $X-\{x_n\}$. 
Now, let $\cT_n$ be a binary phylogenetic $X$-tree, and let $\cT_{n-1}$ be $\cT_n|X'$. Furthermore, let $\cC$ be the subset of $X'$ that consists of the descendant leaves of the parent of $x_n$ in $\cT_n$. As $\cU_{n-1}$ displays $\cT_{n-1}$, there exist an embedding $\cE$ of $\cT_{n-1}$ in $\cU_{n-1}$ and an edge $(u,v)$ in $\cE$ such that the set of descendants of $v$ in $\cE$ is precisely $\cC$. If $(u,v)$ is a tree edge in $\cU_{n-1}$, then it is easily checked that $\cU_n$ displays $\cT_n$ by construction. On the other hand, if $(u,v)$ is a reticulation edge in $\cU_{n-1}$, then $v$ has out-degree one in $\cE$. Let $(v,w)$ be the unique edge in $\cE$ that is directed out of $v$. Note that, as $\cU_{n-1}$ is tree-child, $w$ is a tree vertex in $\cU_{n-1}$. Then, as $(v,w)$ is a tree edge in $\cU_{n-1}$ that is subdivided by a new vertex in the construction of $\cU_n$ from $\cU_{n-1}$, it again follows that $\cU_n$ displays $\cT_n$. This completes the proof of the theorem.\qed
\end{proof}

The next corollary is an immediate consequence of Theorem~\ref{t:all} and the fact that every phylogenetic tree has a binary refinement on the same leaf set.
\begin{corollary}\label{cor:tcdisplay2}
Let $\cP$ be a set of phylogenetic $X$-trees. There exists a tree-child network on $X$ that displays $\cP$.
\end{corollary}

%This corollary sharply contrasts with the analogous result for temporal \plum{tree-child} networks \cite[Theorem 2]{humphries13}. In particular, not every collection $\cP$ of phylogenetic trees can be displayed by a temporal \plum{tree-child} network \plum{whose leaf set is the union of leaf sets in $\cP$.}
%%on $\cL(\cP)$.

While every collection of phylogenetic $X$-trees can be displayed by a tree-child network on $X$, a simple counting argument shows that the analogous result is not true for {\it binary} tree-child networks. Specifically, a binary tree-child network on $X$ has at most $|X|-1$ reticulations~\cite[Proposition 1]{cardona12} and so displays at most $2^{|X|-1}$ distinct binary phylogenetic $X$-trees. But for large enough $X$, there are many more distinct binary phylogenetic $X$-trees than $2^{|X|-1}$. For related results, we refer the interested reader to ~\cite{simpson}.

\section{Cherry-picking characterisations}\label{sec:main}

In this section, we state the two cherry-picking characterisations whose proofs are given in the next section. Let $\cT$ be a  phylogenetic $X$-tree with root $\rho$, where $|X|\ge 2$. If $x$ is a leaf of $\cT$, we denote by $\cT\ba x$ the operation of deleting $x$ and its incident edge and, if the parent of $x$ in $\cT$ has out-degree two, suppressing the resulting degree-two vertex. Note that if the parent of $x$ is $\rho$ and $\rho$ has out-degree two, then $\cT\ba x$ denotes the operation of deleting $x$ and its incident edge, and then deleting $\rho$ and its incident edge. Observe that $\cT\ba x$ is a phylogenetic tree on $X-\{x\}$. A $2$-element subset $\{x, y\}$ of $X$ is a {\it cherry} of $\cT$ if $x$ and $y$ have the same parent. Clearly, every  phylogenetic tree with at least two leaves contains a cherry. In this paper, we typically distinguish the leaves in a cherry, in which case we write $\{x, y\}$ as the ordered pair $(x, y)$ depending on the roles of $x$ and $y$.

%\marginpar{\blue{Is the placement of figure correct? Also refer to it in reference to $h(\cN)$ in the text.}}

Let $\cT$ be a phylogenetic $X$-tree and let $(x, y)$ be an ordered pair of leaves in $X$. If $(x, y)$ is a cherry of $\cT$, then let $\cT'=\cT\backslash x$; otherwise, let $\cT'=\cT$. We say that $\cT'$ has been obtained from $\cT$ by {\em cherry picking $(x, y)$}. Now, let $\cP$ be a set of phylogenetic $X$-trees, and let
$$\sigma=(x_1, y_1), (x_2, y_2), \ldots, (x_s, y_s), (x_{s+1}, -)$$
be a sequence of ordered pairs in $X\times (X\cup \{-\})$ such that the following property is satisfied.\\

\noindent (P) For all $i\in \{1, 2, \ldots, s\}$, we have $x_i\not\in \{y_{i+1}, y_{i+2}, \ldots, y_s\}$.\\

\noindent Setting $\cP_0=\cP$ and, for all $i\in \{1, 2, \ldots, s\}$, setting $\cP_i$ to be the set of phylogenetic trees obtained from $\cP_{i-1}$ by cherry picking $(x_i, y_i)$ in each tree in $\cP_{i-1}$, we call $\sigma$ a {\em cherry-picking sequence of $\cP$} if each tree in $\cP_s$ consists of the single vertex $x_{s+1}$.\\

%Let $\cP$ be a set of phylogenetic $X$-trees. A sequence
%$$\sigma=(x_1, y_1), (x_2, y_2), \ldots, (x_s, y_s), (x_{s+1}, -)$$
%of ordered pairs in $X\times (X\cup \{-\})$ is a {\it cherry-picking sequence of $\cP$} if the following algorithm returns a set of phylogenetic trees each of which consists of a single vertex in $\{x_{s+1}, x_{s+2}, \ldots, x_t\}$.\\
%
%\noindent {\bf Algorithm.} {\sc Picking Cherries} \\
%\noindent{\bf Input.} A set $\cP$ of phylogenetic $X$-trees and a cherry-picking sequence $$\sigma=(x_1, y_1), (x_2, y_2), \ldots, (x_s, y_s), (x_{s+1}, -), (x_{s+2}, -), \ldots, (x_t, -)$$ for $\cP$.\\
%\noindent{\bf Output.} A set $\cP_s$ of  phylogenetic trees.
%\begin{enumerate}[{\bf Step 1.}]
%\item Set $\cP_0=\cP$ and, for each tree $\cT\in \cP$, set $\cT_0=\cT$. Set $i=1$. 
%\item Set $\cP_{i}$ to be the set of phylogenetic trees obtained from $\cP_{i-1}$ by performing exactly one of the following two operations for each tree $\cT_{i-1}\in \cP_{i-1}$:
%\begin{enumerate}[{\bf (a)}]
%\item If $\{x_i, y_i\}$ is a cherry of $\cT_{i-1}$, then set $\cT_{i}=\cT_{i-1}\ba x_i$.
%\item Else, set $\cT_{i}=\cT_{i-1}$.
%\end{enumerate}
%%\item Set $\cP_{i}=\{\cT_{i}: \cT_{i-1}\in \cP_{i-1}\}$.
%\item If $i< s$, increment $i$ by one and repeat Step 2; otherwise, return $\cP_s$.\\
%\end{enumerate}
\noindent Furthermore, for all $i\in \{1, 2, \ldots,s\}$, we say that $\cP_i$ is obtained from $\cP$ by {\it picking $x_1, x_2, \ldots, x_{i}$}. 
Additionally, if $\cP_i\ne\cP_{i+1}$, then we refer to $(x_i,y_i)$ as being {\it essential.}  Moreover, if $\sigma$ is a cherry-picking sequence for $\cP$, then the {\it weight} of $\sigma$, denoted $w(\sigma)$, is the value $s+1-|X|$. Observe that, if $\sigma$ is a cherry-picking sequence of $\cP$, then
$$s+1-|X|\ge 0$$
as each element in $X$ must appear as the first element in an ordered pair in $\sigma$.

%A particular type of cherry-picking sequence underlies our characterisation of $h_{\rm tc}(\cP)$. To this end, let $\cP$ be a set of  phylogenetic $X$-trees. A cherry-picking sequence
%$$\sigma=(x_1, y_1), (x_2, y_2), \ldots, (x_s, y_s), (x_{s+1}, -), (x_{s+2}, -), \ldots, (x_t, -)$$
%for $\cP$ is called a {\it tree-child sequence} if $t=s+1$ and, for all $i\in \{1, 2, \ldots, s\}$, we have $x_i\not\in \{y_{i+1}, y_{i+2}, \ldots, y_s\}$. 
Now, let $\sigma$ be a cherry-picking sequence for $\cP$. We call $\sigma$ a {\it minimum cherry-picking sequence} of $\cP$ if $w(\sigma)$ is of smallest value over all cherry-picking sequences of $\cP$. This smallest value is denoted by $s(\cP)$. It will follow from the results in the next section (Lemma~\ref{tccp3}) that every collection $\cP$ of  phylogenetic trees has a cherry-picking sequence and so $s(\cP)$ is well defined.
%Since every binary phylogenetic tree has a cherry, $s_{\rm tc}(\cP)$ is well defined.
Referring to Figure~\ref{fig:cps}, $$\sigma=(3,2),(3,4),(5,6),(5,4),(1,2),(4,2),(4,6),(2,6),(6,-)$$ is a cherry-picking sequence with weight $w(\sigma)=9-6=3$ for the four trees shown at the top of this figure.\\

\noindent {\bf Remark.} As noted in the introduction, cherry-picking sequences were introduced in~\cite{humphries13}. In the set-up of this paper, the difference is as follows. Instead of a cherry-picking sequence consisting of a set of ordered pairs, a cherry-picking sequence in~\cite{humphries13} consists of an ordering of the elements in $X$. Moreover, this ordering has the additional property that, for each $i\in\{1,2,\ldots,s\}$, $x_i$ is part of a cherry of {\it every} tree in $\cP_{i-1}$. Subsequently, $x_i$ is deleted from each tree in $\cP_{i-1}$, and the iterative process continues. The weighting of such a sequence is based, across  all $i$, on the number of different cherries of which $x_i$ is part of. It is not difficult to see how this could be interpreted as a special type of cherry-picking sequence as defined in this paper.\\ 

The first of our new characterisations is the next theorem. For a given set $\cP$ of  phylogenetic $X$-trees, it writes $h_{\rm tc}(\cP)$ in terms of cherry-picking sequences for $\cP$.

\begin{theorem}
Let $\cP$ be a set of phylogenetic $X$-trees. Then $$h_{\rm tc}(\cP)=s(\cP).$$
\label{picking1}
\end{theorem}

%While we have kept the definition of a cherry-picking sequence as general as possible, we remark that in order to compute a minimum cherry-picking sequence for a set $\cP$ of binary phylogenetic trees in practice, it is sufficient to only consider those cherry-picking sequences for which {\sc Picking Cherries} picks at least one cherry in each iteration of the algorithm i.e. $\cP_{i-1}\ne \cP_i$ for each $i\in\{1,2,\ldots,s\}$.

To state the second characterisation, we require an additional concept. Let $\cT$ be a phylogenetic $X$-tree. Consider the operation of adjoining a new leaf $z$ to $\cT$ in one of the following three ways.

\begin{enumerate}[(i)]
\item Subdivide an edge of $\cT$ with a new vertex, $u$ say, and add the edge $(u, z)$.
\item View the root $\rho$ of $\cT$ as a degree-one vertex adjacent to the original root and add the edge $(\rho, z)$. 
\item Add the edge $(v, z)$, where $v$ is an interior vertex of $\cT$.
\end{enumerate}
\noindent We refer to this operation as {\it attaching a new leaf $z$} to $\cT$. More generally, if $Z$ is a finite set of elements such that $X\cap Z$ is empty, then {\it attaching $Z$} to $\cT$ is the operation of attaching, in turn, each element in $Z$ to $\cT$ to eventually obtain a phylogenetic tree on $X\cup Z$. We refer to $Z$ as a set of {\it auxiliary leaves}. Lastly, {\it attaching Z} to a set $\cP$ of phylogenetic $X$-trees is the operation of attaching $Z$ to each tree in $\cP$.

Let $\cP$ be a set of phylogenetic $X$-trees. 
%A cherry-picking sequence $\sigma$ of $\cP$ is {\it leaf added} if it is a cherry-picking sequencing of a set of phylogenetic trees obtained from $\cP$ by attaching a set of auxiliary leaves.  
A sequence $$\sigma=(x_1, y_1), (x_2, y_2), \ldots, (x_s, y_s), (x_{s+1}, -)$$
of ordered pairs in $(X\cup Z)\times (X\cup Z\cup\{-\})$ that satisfies (P) is a {\it leaf-added cherry-picking sequence} for $\cP$  if it is a cherry-picking sequence of a set of  phylogenetic trees obtained from $\cP$ by attaching $Z$. As for cherry-picking sequences, the {\it weight} of $\sigma$, denoted $w(\sigma)$, is the value $s+1-(|X|+|Z|)$. We denote the minimum weight amongst all leaf-added cherry-picking sequences of $\cP$ by $s_+(\cP)$. Of course, $s_+(\cP)\le s(\cP)$, but this inequality can also be strict. To illustrate, consider the two sets $\cP$ and $\cP'$ of phylogenetic trees shown in Figure~\ref{fig:leaf-added}. Now
\begin{eqnarray}
\sigma&=&(4,5),(4,1),(4,3),(5,6),(5,3),(5,8),\nonumber\\
&&(2,3),(3,1),(6,7),(7,8),(1,8),(8,-)\nonumber
\end{eqnarray}
is a cherry-picking sequence for $\cP$ of weight $w(\sigma)=12-8=4$. In fact, it follows from~\cite{iersel16,kelk12COM} that 
$h_{\rm tc}(\cP)=4$ (see Section~\ref{sec:conclu} for details). On the other hand,
\begin{eqnarray}
\sigma'&=&(5,z),(5,8),(4,z),(4,1),(z,3),(z,6),\nonumber\\
&&(2,3),(3,1),(6,7),(7,8),(1,8),(8,-)\nonumber
\end{eqnarray}
is a cherry-picking sequence for $\cP'$ of weight $w(\sigma')=12-9=3$. 
Since $\cP'$ can be obtained by attaching $z$ to $\cP$, it follows that $\sigma'$ is a leaf-added cherry-picking sequence for $\cP$ and $s_+(\cP)\le 3$.

For a given set $\cP$ of phylogenetic $X$-trees, the next theorem characterises $h(\cP)$ in terms of leaf-added cherry-picking sequences.

\begin{theorem}
Let $\cP$ be a set of  phylogenetic $X$-trees. Then $$h(\cP)=s_+(\cP).$$
\label{picking2}
\end{theorem}

\begin{figure}[t]
\center
\scalebox{1.2}{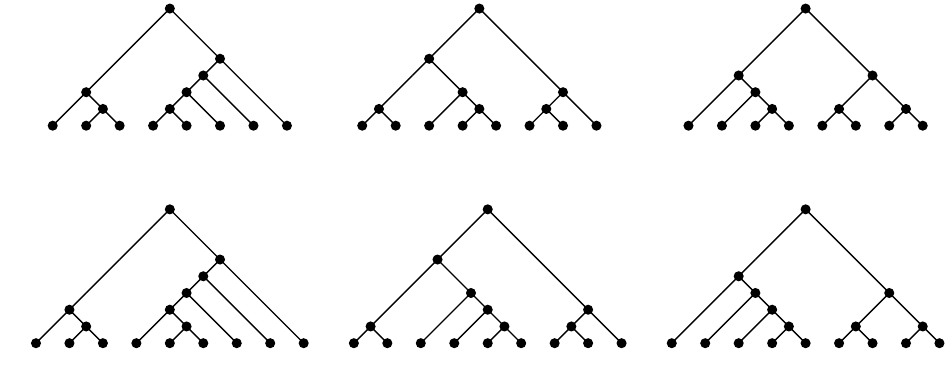}
\caption{Two sets $\cP$ and $\cP'$ of phylogenetic trees, where $\cP'$ is obtained  by attaching $z$ to $\cP$. (In parts, adapted from~\cite[Figure 1]{iersel16}.)}
\label{fig:leaf-added}
\end{figure}

It is worth noting that, for a set $\cP$ of  phylogenetic $X$-trees, it follows from Theorems~\ref{picking1} and~\ref{picking2} that $h_{\rm tc}(\cP)$ and $h(\cP)$ can be determined without constructing a phylogenetic network.

\section{Proofs of Theorems~\ref{picking1} and~\ref{picking2}}
\label{pickingproofs}

In this section, we prove Theorems~\ref{picking1} and~\ref{picking2}. Most of the work is in proving Theorem~\ref{picking1}. We begin by showing that $h_{\rm tc}(\cP)\le s(\cP)$.

\begin{lemma}
Let $\cP$ be a set of  phylogenetic $X$-trees. Let $\sigma$ be a cherry-picking sequence for $\cP$. Then there exists a tree-child network $\cN$ on $X$ that displays $\cP$ with $h(\cN)\le w(\sigma)$ satisfying the following properties:
\begin{enumerate}[{\rm (i)}]
\item If $u$ is a tree vertex in $\cN$ and not a parent of a reticulation, then there are leaves $\ell_1$ and $\ell_2$ at the end of tree paths starting at the children $v_1$ and $v_2$ of $u$, respectively, such that $(\ell_1, \ell_2)$ is an element in $\sigma$.

\item If $u$ is a tree vertex in $\cN$ and a parent of a reticulation $v$, then there are leaves $\ell_u$ and $\ell_v$ at the end of tree paths starting at $u$ and $v$, respectively, such that $(\ell_v, \ell_u)$ is an element in $\sigma$.
\end{enumerate}
\label{tccp1}
\end{lemma}

\begin{proof}
Let
$$\sigma=(x_1, y_1), (x_2, y_2), \ldots, (x_s, y_s), (x_{s+1}, -)$$
be a cherry-picking sequence for $\cP$. The proof is by induction on $s$. If $s=0$, then $|X|=1$ and each tree in $\cP$ consists of the single vertex in $X$. It immediately follows that choosing $\cN$ to be the phylogenetic network consisting of the single vertex in $X$ establishes the lemma for $s=0$.

Now suppose that $s\ge 1$, and that the lemma holds for all cherry-picking sequences for sets of phylogenetic trees on the same leaf set whose length is at most $s$. Let
$$\sigma'=(x_2, y_2), (x_3, y_3), \ldots, (x_s, y_s), (x_{s+1}, -)$$
and let $\cP'$ be the set of phylogenetic trees obtained from $\cP$ by picking $x_1$.

First assume that each tree in $\cP'$ has the same leaf set, namely $X'=X-\{x_1\}$. Then $\sigma'$ is a cherry-picking sequence for $\cP'$. By induction, there is a tree-child network $\cN'$ on $X'$ that displays $\cP'$ with $h(\cN')\le w(\sigma')$ and satisfies (i) and (ii). Since each tree in $\cP'$ has the same leaf set, $\{x_1, y_1\}$ is a cherry in each tree in $\cP$. Therefore, as $\cN'$ displays a binary refinement of each tree in $\cP'$, the tree-child network obtained from $\cN'$ by subdividing the edge directed into $y_1$ with a new vertex $u$ and adding the edge $(u, x_1)$ displays $\cP$. Furthermore, as $h(\cN')\le w(\sigma')$ and $\cN'$ satisfies (i) and (ii) relative to $\sigma'$, we have $h(\cN)=h(\cN')\le w(\sigma')=w(\sigma)$ and it is easily seen that $\cN$ satisfies (i) and (ii) relative to $\sigma$.

Now assume that not every tree in $\cP'$ has the same leaf set.  Let $\cP'_1$ denote the subset of trees in $\cP'$ whose leaf set is $X-\{x_1\}$. Since $\cP'-\cP'_1$ is non-empty, there exists some $i$ with $i\in\{2,3,\ldots,s+1\}$ such that $x_i=x_1$. Note that $(x_1, -)$ is not in $\sigma$; otherwise there is an ordered pair in $\sigma$ whose second coordinate is $x_1$ and so $\sigma$ is not a cherry-picking sequence for $\cP$. Let $(x_1, y_i)$ be the first ordered pair in $\sigma'$ whose first coordinate is $x_1$. Let $\cT_1$ be a tree in $\cP'_1$. Consider the process of picking, in order, $(x_2,y_2),(x_3,y_3),\ldots,(x_{i-1},y_{i-1})$ from $\cT_1$.
%and, using $\sigma'$, consider applying $i-2$ iterations of {\sc Picking Cherries} to $\cT_1$. 
Let $X_1$ denote the subset of leaves in $X-\{x_1\}$ that are deleted from $\cT_1$ in this process. Observe that, as $y_i$ is the second coordinate in $(x_i, y_i)$, we have $y_i\not\in X_1$.

We next add $x_1$ to $\cT_1$ to obtain a phylogenetic $X$-tree for which $\sigma'$ is a cherry-picking sequence. Let $w$ be the (unique) vertex of $\cT_1$ that is closest to the root with the property that $y_i$ is a descendant leaf of $w$, and the child  of $w$ on the path from $w$ to $y_i$ has all its descendant leaves in $X_1\cup \{y_i\}$.
%that its descendant leaves contain $y_i$ and all remaining descendant leaves are in $X_1$. 
Let $\cT'_1$ be the phylogenetic $X$-tree obtained from $\cT_1$ by
%subdividing the edge directed into $w$ with a new vertex $u$ and 
adding the edge $(w, x_1)$. We now show that $\sigma'$ is a cherry-picking sequence for $\cT'_1$. Suppose that $\sigma'$ is not a cherry-picking sequence for $\cT'_1$. Let $u$ be the parent of $w$ in $\cT_1'$. Then amongst the first $i-2$ ordered pairs in $\sigma'$ is an ordered pair of the form $(x_j, y_i)$ 
%\marginpar{Is index $j$ correct here? Need to check up to the end of this proof. We use $(x_i,y_i)$ above.} 
that is essential when, the ordered pairs in $\sigma'$ are (in order) picked from $\cT_1$,
%using $\sigma'$, {\sc Picking Cherries} is applied to $\cT_1$,
where $x_j$ is a descendant leaf of $u$ in $\cT_1'$. But then, each descendant leaf of $w$ is in $X_1\cup\{y_i\}$, contradicting the choice of $w$.

Repeating this placement of $x_1$ for each tree in $\cP'_1$, we obtain a set $\cP''_1$ of phylogenetic $X$-trees from $\cP'_1$. Let $\cP''=\cP''_1\cup (\cP'-\cP'_1)$ and observe that $\sigma'$ is a cherry-picking sequence for $\cP''$. Therefore, by induction, there is a tree-child network $\cN'$ on $X$ that displays $\cP''$ with $h(\cN')\le w(\sigma')$ and satisfies (i) and (ii).

Let $p$ denote the parent of $x_1$ in $\cN'$. If $p$ is a reticulation, let $\cN$ be the phylogenetic network obtained from $\cN'$ by subdividing the edge directed into $y_1$ with a new vertex $u$ and adding the edge $(u, p)$. Since $\cN'$ is tree-child and displays $\cP'$, it follows that $\cN$ is tree-child and displays $\cP$. Furthermore,
$$h(\cN)=h(\cN')+1\le w(\sigma')+1=w(\sigma).$$
Additionally, as $(x_1, y_1)\in \sigma$, it also follows that, as $\cN'$ satisfies (i) and (ii) relative to $\sigma'$, we have $\cN$ satisfies (i) and (ii) relative to $\sigma$.

Thus we may assume that $p$ is a tree vertex. Let $w$ denote the child of $p$ that is not $x_1$ in $\cN'$. If $w$ is a reticulation, then, as $\cN'$ satisfies (ii), $\sigma'$ contains a cherry in which $x_1$ is the second coordinate. But $(x_1, y_1)$ is the first ordered pair in $\sigma$ and so, as $\sigma$ satisfies (P), $x_1$ is never the second coordinate in an ordered pair in $\sigma$; a contradiction. Therefore $w$ is either a tree vertex or a leaf in $\cN'$. So, as $\cN'$ satisfies (i) and no ordered pair has $x_1$ as the second coordinate, it follows that $\sigma'$ contains an ordered pair, $(x_1, y_j)$ say, where $y_j$ is the leaf at the end of a tree path in $\cN'$ starting at $w$. Now let $\cN$ be the phylogenetic network obtained from $\cN'$ by subdividing the edges directed into $y_1$ and $x_1$ with new vertices $u$ and $v$, respectively, and adding the edge $(u, v)$. Since $\cN'$ is tree-child and $h(\cN')\le w(\sigma')$, it is easily seen that $\cN$ is tree-child and $h(\cN)=h(\cN')+1\le w(\sigma')+1=w(\sigma)$. Furthermore, $\cN'$ displays $\cP'-\cP'_1$ as well as $\cP''_1$, and therefore $\cP'_1|(X-\{x_1\})$. Thus $\cN$ displays $\cP$. To see that $\cN$ satisfies (i) and (ii) relative to $\sigma$, it suffices to show that $\cN$ satisfies (ii) for $p$ and $u$. Indeed, the two ordered pairs $(x_1, y_j)$ and $(x_1, y_1)$ in $\sigma$ verify (ii) for $p$ and $u$, respectively. This completes the proof of the lemma.\qed
\end{proof}

The next corollary immediately follows from Lemma~\ref{tccp1}.

\begin{corollary}
Let $\cP$ be a set of phylogenetic $X$-trees. Then $h_{\rm tc}(\cP)\le s(\cP)$.
\label{tccp2}
\end{corollary}

For the proof of the converse of Corollary~\ref{tccp2}, we begin with an additional lemma. Let $\cN$ be a phylogenetic network, and let $x$ and $y$ be two leaves in $\cN$. Generalising cherries to phylogenetic networks, we say that $\{x, y\}$ is a {\it cherry} in $\cN$ if $x$ and $y$ have a common parent. Moreover, we call $\{x, y\}$ a {\it reticulated cherry} if the parent of $x$, say $p_x$, and the parent of $y$, say $p_y$, are joined by a reticulation edge $(p_y, p_x)$ in which case we say that $x$ is the {\it reticulation leaf} relative to $\{x,y\}$. We next define two operations on $\cN$. First, {\it reducing a cherry $\{x,y\}$} is the operation of deleting one of the two leaves in $\{x,y\}$, and suppressing the resulting degree-two vertex. Second, {\it reducing a reticulated cherry $\{x,y\}$} is the operation of deleting the reticulation edge joining the parents of $x$ and $y$ and suppressing any resulting degree-two vertices. The proof of the next lemma is similar to the analogous result for binary tree-child networks [5, Lemma 4.1] and is omitted.

\begin{lemma}\label{l:cherry}
Let $\cN$ be a tree-child network on $X$. Then the following hold.
\begin{enumerate}[{\rm (i)}]
\item  If $|X|\geq 2$, then $\cN$ contains either a cherry or a reticulated cherry.
\item If $\cN'$ is obtained from $\cN$ by reducing either a cherry or a reticulated cherry, then $\cN'$ is a tree-child network.
\end{enumerate}
\end{lemma}

\begin{lemma}\label{tccp3}
Let $\cP$ be a set of  phylogenetic $X$-trees. Then $h_{\rm tc}(\cP)\ge s(\cP)$.
\end{lemma}

\begin{proof}
Let $\cN$ be a tree-child network on $X$ that displays $\cP$. By Corollary~\ref{cor:tcdisplay2}, such a network exists.  We establish the lemma by explicitly constructing a cherry-picking sequence $\sigma$ for $\cP$ such that $w(\sigma)\le h(\cN)$.

Let $\rho$ denote the root of $\cN$, and let $v_1, v_2, \ldots, v_r$ denote the reticulations of $\cN$. Let $\ell_{\rho}, \ell_1, \ell_2, \ldots, \ell_r$ denote the leaves at the end of tree paths $P_{\rho}, P_1, P_2, \ldots, P_r$ in $\cN$ starting at $\rho, v_1, v_2, \ldots, v_r$, respectively. Observe that these paths are pairwise vertex disjoint. We now construct a sequence of ordered pairs as follows:\\
\begin{enumerate}[{\bf Step 1.}]
\item Set $\cN=\cN_0$ and $\sigma_0$ to be the empty sequence. Set $i=1$.

\item If $\cN_{i-1}$ consists of a single vertex $x_i$, then set $\sigma_i$ to be the concatenation of $\sigma_{i-1}$ and $(x_i, -)$, and return $\sigma_i$.

\item If $\{x_i, y_i\}$ is a cherry in $\cN_{i-1}$, then
\begin{enumerate}[{\bf (a)}]
\item If one of $x_i$ and $y_i$, say $x_i$, equates to $\ell_j$ for some $j\in \{1, 2, \ldots, r\}$ and $v_j$ is not a reticulation in $\cN_{i-1}$, then set $\sigma_i$ to be the concatenation of $\sigma_{i-1}$ and $(x_i, y_i)$.

\item Otherwise, set $\sigma_i$ to be the concatenation of $\sigma_{i-1}$ and $(x_i, y_i)$, where $x_i\not\in\ \{\ell_{\rho}, \ell_1, \ell_2, \ldots, \ell_r\}$.

\item Set $\cN_i$ to be the tree-child network obtained from $\cN_{i-1}$ by deleting $x_i$, thereby reducing the cherry $\{x_i, y_i\}$.

\item Increase $i$ by one and go to Step 2.
\end{enumerate}

\item Else, there is a reticulated cherry $\{x_i, y_i\}$ in $\cN_{i-1}$, where $x_i$ say is the reticulation leaf.
\begin{enumerate}[{\bf (a)}]
\item Set $\sigma_i$ to be the concatenation of $\sigma_{i-1}$ and $(x_i, y_i)$.

\item Set $\cN_i$ to be the tree-child network obtained from $\cN_{i-1}$ by reducing the reticulated cherry $\{x_i,y_i\}$.
%deleting the reticulation edge of $\{x_i, y_i\}$ and suppressing the resulting degree-two vertices.

\item Increase $i$ by one and go to Step 2.\\
\end{enumerate}
\end{enumerate}

First note that it is easily checked that the construction is well defined, that is, it returns a sequence of ordered pairs. Moreover, in each iteration $i$ of the above construction, it follows from Lemma~\ref{l:cherry} that $\cN_i$ is tree-child. We next show that, if $\{x_i,y_i\}$ is a cherry in $\cN_{i-1}$, and $x_i$ and $y_i$ equate to $\ell_j$ and $\ell_{j'}$, respectively, where $\ell_j$ and $\ell_{j'}$ are elements in $\{\ell_1,\ell_2,\ldots,\ell_r\}$, then exactly one of $v_j$ and $v_{j'}$ is a reticulation in $\cN_{i-1}$. To see this, if $v_j$ and $v_{j'}$ are both reticulations in $\cN_{i-1}$, then $P_j$ and $P_{j'}$ are not vertex disjoint in $\cN$; a contradiction. On the other hand, suppose neither $v_j$ and $v_{j'}$ are reticulations in $\cN_{i-1}$. Without loss of generality, we may assume $\{x_i, y_i\}$ is the first such cherry for which this holds. Since $\cN$ is tree-child, and therefore has no tree vertex that is the parent of two reticulations, there is an iteration $i'< i$, in which the cherry $(x_{i'}, y_{i'})$ is concatenated with $\sigma_{i'-1}$, where $y_{i'}\in \{x_i, y_i\}$, and $\cN_{i'}$ has $\{x_i, y_i\}$ as a cherry but $\cN_{i'-1}$ does not. If $x_{i'}=\ell_{\rho}$ or $x_{i'}\in \{\ell_1, \ell_2, \ldots, \ell_r\}$,
we contradict the construction by the choice of $\{x_i, y_i\}$. Also, if $x_{i'}\not\in \{\ell_{\rho}, \ell_1, \ell_2, \ldots, \ell_r\}$, then we again contradict the construction. Hence, we may assume for the remainder of the proof that exactly one of $v_j$ and $v_{j'}$ is a reticulation in $\cN_{i-1}$.

Let
$$\sigma=(x_1, y_1), (x_2, y_2), \ldots, (x_{i-1}, y_{i-1}), (x_i, -)$$
be the sequence returned by the construction. We prove by induction on $i$ that $\sigma$ is a cherry-picking sequence for $\cP$ whose weight is at most $h(\cN)$. If $i=1$, then $\cN$ consists of the single vertex in $X$ and the construction correctly returns such a sequence.

Now suppose that $i\ge 2$, and consider the first iteration of the construction. Either $\{x_1, y_1\}$ is a cherry or a reticulated cherry of $\cN_0$. If $\{x_1, y_1\}$ is a cherry, then $\{x_1, y_1\}$ is a cherry of each tree in $\cP$. In this instance, let $\cP'$ denote the set of phylogenetic $X'$-trees obtained from $\cP$ by picking $x_1$, where $X'=X-\{x_1\}$. Observe that $\cN_1$ is a tree-child network on $X'$ that displays $\cP'$.

Now assume that $\{x_1, y_1\}$ is a reticulated cherry with $x_1$ as the reticulation leaf. Let $\cP_1$ be the subset of trees in $\cP$ not displayed by $\cN_1$ and let $\cP_2=\cP-\cP_1$. Note that $\{x_1, y_1\}$ is a cherry of each tree in $\cP_1$. For each tree in $\cP_1$, delete the edge incident with $x_1$, suppress any resulting degree-two vertex, and reattach $x_1$ to the rest of the tree containing $y_1$ by subdividing an edge with a new vertex and adding an edge joining this vertex and $x_1$ so that the resulting  phylogenetic $X$-tree is displayed by $\cN_1$. It is easily seen that this is always possible. Let $\cP'_1$ denote the resulting collection of trees obtained from $\cP_1$. For this instance, let $\cP'=\cP'_1\cup \cP_2$ and observe that $\cN_1$ displays $\cP'$. 

To complete the induction it suffices to show that if
$$\sigma'=(x_2, y_2), (x_3, y_3), \ldots, (x_{i-1}, y_{i-1}), (x_i, -)$$
is a cherry-picking sequence for $\cP'$ whose weight $w(\sigma')$ is at most $h(\cN_1)$, then $\sigma$ is a cherry-picking sequence for $\cP$ whose weight $w(\sigma)$ is at most $h(\cN_0)$, that is, at most $h(\cN)$. First assume that $\{x_1, y_1\}$ is a cherry of $\cN_0$. Then, as $\sigma'$ satisfies (P) and $x_1\not\in \cL(\cN_1)$, it follows that $\sigma$ also satisfies (P) and so $\sigma$ is a cherry-picking sequence for $\cP$. Since $x_1$ only appears once as the first coordinate of an ordered pair in $\sigma$, we have
$$w(\sigma)=w(\sigma')\le h(\cN_1)=h(\cN_0).$$

Now assume that $\{x_1, y_1\}$ is a reticulated cherry of $\cN_0$ with $x_1$ as the reticulation leaf. Without loss of generality, let $v_1$ denote the associated reticulation, so that $x_1=\ell_1$. We next show that $\sigma$ satisfies (P). If the in-degree of $v_1$ is at least three in $\cN_0$, then $v_1$ exists in $\cN_1$ and so, by construction, $x_1$ does not appear as the second coordinate of an ordered pair in $\sigma'$ as well as in $\sigma$. Therefore, if the in-degree of $v_1$ is at least three in $\cN_0$, then $\sigma$ satisfies (P).

Now suppose that  the in-degree of $v_1$ is two in $\cN_0$. To establish that $\sigma$ satisfies (P), assume to the contrary that $x_1$ appears as the second coordinate of an ordered pair in $\sigma'$. Let $(z, x_1)$ denote the first such ordered pair. Then, at some iteration $j$, either $\{z, x_1\}$ is a cherry or a reticulated cherry of $\cN_{j-1}$. If $\{z, x_1\}$ is a cherry of $\cN_{j-1}$, then, since $x_1=\ell_1$, we are in Step 3(a) in iteration $j$ of the construction and so the ordered pair should be $(x_1, z)$; a contradiction. On the other hand, if $\{z, x_1\}$ is a reticulated cherry of $\cN_{j-1}$, then $z$ is the reticulation leaf of $\{z, x_1\}$ and, by construction of $\sigma$, one of the parents of $v_1$ in $\cN_0$ is the parent of two reticulations in $\cN_0$, namely $v_1$ and the reticulation for which, by construction, there is a tree path starting at this reticulation and ending at $z$; a contradiction as $\cN_0$ is tree-child. Hence $\sigma$ satisfies (P). Thus, since each tree in $\cP_1$ has $\{x_1, y_1\}$ as a cherry and $\sigma'$ is a cherry-picking sequence for $\cP'$, it follows that $\sigma$ is a cherry-picking sequence for $\cP$. Furthermore, as $w(\sigma)=w(\sigma')+1$ and $h(\cN_0)=h(\cN_1)+1$,
$$w(\sigma)=w(\sigma')+1\le h(\cN_1)+1=h(\cN_0).$$
This completes the proof of the lemma.\qed
\end{proof}

\noindent {\sc Proof of Theorem~\ref{picking1}.}
Combining Corollary~\ref{tccp2} and Lemma~\ref{tccp3} establishes the theorem.\qed\\

We next establish Theorem~\ref{picking2}.\\

\noindent{\sc Proof of Theorem~\ref{picking2}.}
We first show that $h(\cP)\le s_+(\cP)$. Let $\cP'$ be a set of phylogenetic trees obtained from $\cP$ by attaching a set $Z$ such that $X\cap Z$ is empty and $s_+(\cP)=s(\cP')$. It follows by Theorem~\ref{picking1} that there is a tree-child network $\cN'$ on $X\cup Z$ that displays $\cP'$ with $h(\cN')=s(\cP')$. Observe that $\cN'$ displays $\cP$. Let $\cN$ be the phylogenetic network on $X$ obtained from $\cN'$ by deleting every vertex that is not on a directed path from the root to a leaf in $X$, and suppressing any resulting non-root vertex of degree two. Noting that no deleted vertex is used to display a  phylogenetic tree in $\cP$, it is easily checked that, up to the root having out-degree one, $\cN$ displays $\cP$. Furthermore, $h(\cN)\le h(\cN')$. Therefore, by Theorem~\ref{picking1},
$$h(\cP)\le h(\cN)\le h(\cN')=s(\cP')=s_+(\cP).$$
In particular, $h(\cP)\le s_+(\cP)$.

To prove the converse, $h(\cP)\ge s_+(\cP)$, let $\cN$ be a phylogenetic network on $X$ that displays $\cP$ and $h(\cN)=h(\cP)$. Let $\cN'$ be the phylogenetic network obtained by attaching a new leaf to each reticulation edge in $\cN$, i.e. for each reticulation edge $e$, subdivide $e$ with a new vertex $u$ and add a new edge $(u, z_e)$, where $z_e\notin X$. It is easily checked that $\cN'$ is tree-child and $h(\cN')=h(\cN)$. Let $Z$ denote the set of new leaves attached to $\cN$. For each tree $\cT$ in $\cP$, let $\cT_r$ denote a binary refinement of $\cT$ that is displayed by $\cN$, and let $\cT_r'$ be a binary phylogenetic tree with leaf set $X\cup Z$ that is displayed by $\cN'$ and obtained from $\cT_r$ by attaching $Z$. Note that $\cT'_r$ is a binary refinement of a tree that can be obtained from $\cT$ by attaching $Z$. Set
$$\cP_r'=\{\cT_r': \cT\in \cP\},$$
and note that $h_{\rm tc}(\cP_r')\le h(\cN')$ as $\cN'$ is tree-child and displays $\cP_r'$. Since each tree in $\cP_r'$ is a binary refinement of a tree that can be obtained from a tree in $\cP$ by attaching $Z$, we have $s_+(\cP)\le s(\cP_r')$. Thus, by Theorem~\ref{picking1},
$$s_+(\cP)\le s(\cP_r')=h_{\rm tc}(\cP_r')\le h(\cN')=h(\cN)=h(\cP),$$
and so $s_+(\cP)\le h(\cP)$. \qed\\

We end this section with the pseudocode of an algorithm---called {\sc Construct Tree-Child Network}---that constructs a tree-child network from a cherry-picking sequence. Specifically, given a cherry-picking sequence $\sigma$ for a set $\cP$ of phylogenetic $X$-trees, {\sc Construct Tree-Child Network} returns a tree-child network $\cN$ on $X$ that displays $\cP$ and $h(\cN)\leq w(\sigma)$. This is the same construction as that used to prove Lemma~\ref{tccp1} and so the proof of its correctness is not given.\\

\noindent{\bf Algorithm.} {\sc Construct Tree-Child Network} \\ 
\noindent{\bf Input.} A set $\cP$ of  phylogenetic $X$-trees, and a cherry-picking sequence $$\sigma=(x_1, y_1), (x_2, y_2), \ldots, (x_s, y_s), (x_{s+1}, -)$$
for $\cP$.\\
\noindent{\bf Output.} A tree-child network $\cN$ on $X$ that displays $\cP$ and $h(\cN)\leq w(\sigma)$.
\begin{enumerate}[{\bf Step 1.}]
\item If $|X|=1$, set $\cN_{s+1}$ to be the phylogenetic network consisting of the single vertex $x_{s+1}$ in $X$ and return $\cN_{s+1}$. Otherwise, set  $\cN_{s+1}$ to be the phylogenetic network consisting of the single edge $(\rho,x_{s+1})$ and set $i=s$.

\item \label{step} Depending on which holds, do exactly one of the following three steps.
\begin{enumerate}[{\bf (a)}]

\item If $x_i\in\cL(\cN_{i+1})$ and the parent $p_i$ of $x_i$ is a reticulation in $\cN_{i+1}$, then obtain $\cN_i$ from $\cN_{i+1}$ by subdividing the edge directed into $y_i$ with a new vertex $u$ and adding a new edge $(u,p_i)$.

\item If $x_i\in\cL(\cN_{i+1})$ and the parent $p_i$ of $x_i$ is not a reticulation in $\cN_{i+1}$, then obtain $\cN_i$ from $\cN_{i+1}$ by subdividing the edge directed into $y_i$ with a new vertex $u$, subdividing the edge $(p_i,x_i)$ with a new vertex $v$, and adding a new edge $(u,v)$.

\item Else $x_i\notin\cL(\cN_{i+1})$, and obtain $\cN_i$ from $\cN_{i+1}$ by subdividing the edge directed into $y_i$ with a new vertex $u$ and adding a new edge $(u,x_i)$.

\end{enumerate}

\item If $i=1$, then set $\cN$ to be the network obtained from $\cN_i$ by deleting the unique edge incident with $\rho$ and return $\cN$. Otherwise, decrement $i$ by one and go to Step~\ref{step}.\\
\end{enumerate}

%Informally and in the language of the above algorithm, the correctness of {\sc Construct Tree-Child Network} can be verified in the following way.  By the definition of a tree-child sequence, $\cN_{s+1}$ displays the unique tree that is returned from calling {\sc Picking Cherries} for $\cP$ and $\sigma$.  Furthermore, it follows from the constructive proof of Lemma~\ref{tccp1} that, after each iteration $i$, $\cN_i$ is tree child and displays the set of binary phylogenetic trees obtained from $\cP$ by picking $x_1,x_2,\ldots,x_{i-1}$. Hence, $\cN$ is tree child and displays $\cP$.~To see that $h(\cN)\leq w(\sigma)$, recall that $h(\cN_{s+1})=0$ and that the number of reticulations is only increased by one in Step~\ref{step} (a) and (b).

Now, let $\sigma$ be a leaf-added cherry-picking sequence for a set $\cP$ of  phylogenetic $X$-trees. Then there exists a set $\cP'$ of  phylogenetic trees on $X\cup Z$ obtained from $\cP$ by attaching $Z$ such that $\sigma$ is a cherry-picking sequence for $\cP'$. It is straightforward to check that the network $\cN$ on $X$ resulting from calling {\sc Construct Tree-Child Network} for $\cP'$ and $\sigma$ and, subsequently, restricting to vertices and edges on a path from the root to leaves in $X$ as described in the first direction of the proof of Theorem~\ref{picking2} displays $\cP$ and $h(\cN)\leq w(\sigma)$.

\section{Bounding the maximum number of auxiliary leaves}\label{sec:bound}

In light of Theorem~\ref{picking2}, a natural question to ask is how many auxiliary leaves need to be attached to a given set $\cP$ of  phylogenetic $X$-trees in order to calculate $h(\cP)$. Attaching auxiliary leaves to $\cP$ is necessary whenever $h(\cP)<h_{\rm tc}(\cP)$. Here, we provide an upper bound on the number of auxiliary leaves in terms of $h_{\rm tc}(\cP)$. We start by introducing two operations that, repeatedly applied, transform any phylogenetic network $\cN$ that displays $\cP$ into a tree-child network without increasing $h(\cN)$ and that displays a set of phylogenetic trees obtained from $\cP$ by attaching auxiliary leaves.

Let $\cN$ be a phylogenetic network on $X$, and let $(u,v)$ be an edge in $\cN$ such that $u$ and $v$ are reticulations. Obtain a phylogenetic network $\cN'$ from $\cN$ by contracting $(u,v)$ and, for each resulting pair of parallel edges, repeatedly deleting one of the two edges in parallel and suppressing the resulting degree-two vertex. We say that $\cN'$ has been obtained from $\cN$ by a {\it contraction}. 

\begin{lemma}\label{l:stack}
Let $\cP$ be a collection of phylogenetic $X$-trees, and let $\cN$ be a phylogenetic network on $X$ that displays $\cP$. Let $\cN'$ be a phylogenetic network obtained from $\cN$ by a contraction. Then $\cN'$ displays $\cP$ and $h(\cN')\le h(\cN)$.
\end{lemma}

\begin{proof}
Let $(u,v)$ be the edge in $\cN$ that is incident with two reticulations and contracted in the process of obtaining $\cN'$ from $\cN$. Furthermore, let $w$ be the vertex in $\cN'$ that results from identifying $u$ and $v$. We have $d^-(w)\leq d^-(u)+d^-(v)-1$ while all other reticulations $w'$ in $\cN'$ correspond to a reticulation $u'$ in $\cN$ and $d^-(w')=d^-(u')$. It now follows that $h(\cN')\le h(\cN)$. 

Now, let $\cT$ be a tree in $\cP$. If $v$ is not used to display $\cT$ in $\cN$, then $u$ is also not used to display $\cT$ in $\cN$, and it is easily seen that $\cN'$ displays $\cT$ without using $w$. On the other hand, if $v$ is used to display $\cT$ in $\cN$, then exactly one parent, $t$ say, of $v$ is used to display $\cT$ in $\cN$. If $t\neq u$, it is clear that $\cN'$ displays $\cP$. {Furthermore, if $t=u$, then exactly one parent of $u$, say $s$, is used to display $\cT$ in $\cN$. Now, regardless of whether or not $s$ is also a parent of $v$ in which case $s$ is suppressed in obtaining $\cN'$ from $\cN$, it again follows that $\cN'$ displays $\cT$.}
%Now, let $\cT$ be an element in $\cP$. To see that $\cT$ is displayed by $\cN'$, note that each path in $\cN$ from a parent of $u$ to $v$ and from a parent of $v$, except $u$, to $v$ corresponds to an edge in $\cN'$ that is directed into $w$. As $\cN$ displays $\cT$, it is easily checked that $\cN'$ displays $\cT$. 
%note that, if an embedding $\cE_\cT$ \marginpar{def. for embedding?} of $\cT$ in $\cN$ uses $(u,v)$ it also uses an edge $(t,u)$ that is  directed into $u$ which corresponds to one of the edges that is directed into $w$ in $\cN'$ and, so, $\cN'$ displays $\cT$. On the other hand, if $\cE_\cT$ does not use $(u,v)$, then it uses an edge 
%Let $\cE_\cT$ be a directed acyclic digraph obtained from of $\cN$ by deleting edges and non-root vertices such that suppressing vertices with degree two in $\cE_\cT$ results in $\cT$. Since $\cT$ is displayed by $\cN$, $\cE_\cT$ exists. If $(u,v)$ is not an edge in $\cE_\cT$, then $\cE_\cT$ does not contain an edge that is directed into $u$ in which case $\cE_\cT$ is also a directed acyclic digraph obtained from $\cN'$ by deleting edges and non-root vertices and, so, $\cN'$ displays $\cT$.
Hence $\cN'$ displays each tree in $\cP$ and the lemma follows.\qed
\end{proof}

We call a phylogenetic network with no edges whose end vertices are both reticulations {\it stack free}. It follows from repeated applications of Lemma~\ref{l:stack} that if $\cP$ is a collection of phylogenetic $X$-trees, then there is a stack-free network $\cN$ on $X$ that displays $\cP$ such that $h(\cN)=h(\cP)$.

For the second operation, let $\cN$ be a phylogenetic network on $X$, and let $u$ be a tree vertex in $\cN$ whose two children are both reticulations. Furthermore, let $(u,v)$ be a reticulation edge, and let $z\notin X$. Obtain a phylogenetic network $\cN'$ on $X\cup\{z\}$ from $\cN$ by subdividing the edge $(u,v)$ with a new vertex $w$ and adding a new edge $(w,z)$. We say that $\cN'$ has been obtained from $\cN$ by a {\it leaf-attaching} operation. 
%The next lemma whose proof is omitted is an immediate consequence of the definition of a leaf-attaching operation.
%\begin{lemma}
%Let $\cP$ be a collection of binary phylogenetic $X$-trees, and let $\cN$ be a phylogenetic network on $X$ displaying $P$. Let $\cN'$ be the phylogenetic network obtained from $\cN$ by a leaf-attaching operation. Then $\cN'$ displays $\cP$ and $h(\cN')=h(\cP)$.
%\marginpar{relaxed version of displaying is used}
%\end{lemma}
\noindent Figure~\ref{fig:bounding-Z} illustrates (a) a contraction and (b) a leaf-attaching operation.

\begin{figure}[t]
\center
\scalebox{1.3}{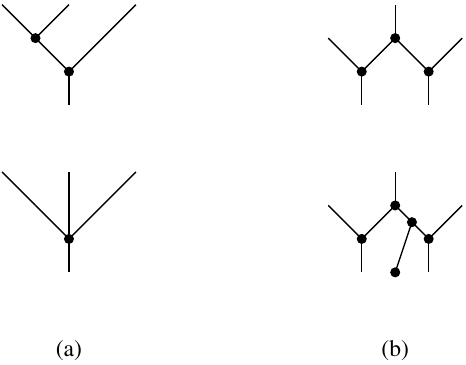}
\caption{The phylogenetic networks at the bottom are obtained from their respective networks at the top by (a) a contraction and (b) a leaf-attaching operation.}
\label{fig:bounding-Z}
\end{figure}

We are now in a position to establish the main result of this section.

\begin{theorem}\label{t:bound}
%Let $\cP$ be a set of binary phylogenetic $X$-trees. There exists a set $\cP_Z$ of binary phylogenetic trees with label set $X\cup Z$ obtained from $\cP$ by attaching $Z$ to $\cP$ such that $h(\cP)=h_{\rm tc}(\cP_Z)$ and $ |Z|\leq h_{\rm tc}(\cP)$.
Let $\cP$ be a collection of phylogenetic $X$-trees. There exists a set $Z$ of auxiliary leaves  with the following two properties.
\begin{itemize}
\item [{\rm (i)}] $ |Z|\leq h_{\rm tc}(\cP)$, and 
\item [{\rm (ii)}] there is a collection $\cP_Z$ of phylogenetic trees with leaf set $X\cup Z$  obtained from $\cP$ by attaching $Z$ such that $h(\cP)=h_{\rm tc}(\cP_Z)$.
\end{itemize}
\end{theorem}

\begin{proof}
Let $\cN$ be a stack-free network on $X$ that displays $\cP$ with $h(\cN)=h(\cP)$. By Lemma~\ref{l:stack}, $\cN$ exists. Now obtain a phylogenetic network $\cN_Z$ from $\cN$ by a minimum number of
%Let $\cN$ be a phylogenetic network on $X$ with $h(\cN)=h(\cP)$. First, obtain a phylogenetic network $\cN'$ from $\cN$ by repeated contractions until the resulting network is stack free. By Lemma~\ref{l:stack}, note that $\cN'$ displays $\cP$ and $h(\cN)=h(\cN')$. Second obtain a network $\cN_Z$ from $\cN'$ by
repeated applications of the leaf-attaching operation until each tree vertex in the resulting network has at least one child that is a tree vertex or a leaf. Clearly, $h(\cN)=h(\cN_Z)$. Moreover, since no leaf-attaching operation results in a new edge in $\cN_Z$ that is incident with two reticulations, $\cN_Z$ is stack free. It now follows that $\cN_Z$ is tree-child. 
Let $Z=\cL(\cN_Z)-X$. Then, by construction, the size of $Z$ is equal to the number of tree vertices in $\cN$ whose two children are both reticulations. 
Let $\cP_Z$ be a set of phylogenetic trees obtained from $\cP$ by attaching $Z$ to $\cP$ such that $\cN_Z$ displays $\cP_Z$. Since $\cN$ displays $\cP$, such a set $\cP_Z$ always exists. By construction, $h(\cP)\geq h_{\rm tc}(\cP_Z)$. Moreover, as each tree in $\cP$ is a restriction of a tree in $\cP_Z$, it follows that $h(\cP)\leq h_{\rm tc}(\cP_Z)$; thereby establishing part (ii) of the theorem.

Using the construction of the previous paragraph, we now establish part (i) of the theorem. Let $E_r$ be the set of reticulation edges in $\cN$, and let $V_t$ be the set of tree vertices of $\cN$ whose children are both reticulations. Recall that $|V_t|=|Z|$. We next make two observations. First, each vertex in $V_t$ is incident with two edges in $E_r$. Second, each edge in $E_r$ is incident with at most one vertex in $V_t$. In summary, this implies that $|Z|\leq \frac 1 2 |E_r|$. Furthermore, we have $|E_r|=h(\cP)+|V_r|$, where $V_r$ is the set of reticulations in $\cN$. Therefore, as $|V_r|\le h(\cP)$, we have $|E_r|\le 2h(\cP)$. As $h(\cP)\leq h_{\rm tc}(\cP)$, it now follows that 
$$|Z|\leq {\textstyle \frac 1 2} |E_r|\leq {\textstyle \frac 1 2} \cdot 2 h(\cP)\leq {\textstyle \frac 1 2} \cdot 2 h_{\rm tc}(\cP)=h_{\rm tc}(\cP).$$
This establishes part (i) of the theorem.\qed
\end{proof}

\section{Scoring an optimum forest}\label{sec:scoring}

%In this section, we consider collections of binary phylogenetic trees. We will show that a particular decision problem---called {\sc Scoring Optimum Forest}---is NP-complete and, so, by establishing NP-completeness for when the input consists, among others, of a set of binary phylogenetic trees also establishes NP-completeness for a set of arbitrary phylogenetic trees.

For a collection $\cP$ of binary phylogenetic $X$-trees, acyclic-agreement forests characterise $h(\cP)$ for when $\cP$ consists of exactly two trees. Indeed, many algorithms and theoretical results that deal with {\sc Minimum Hybridisation} for two trees are deeply-anchored in the notion of acyclic-agreement forests~\cite{albrecht12,baroni05,bordewich07b,kelk12}. In this section, we establish a particular hardness result that contributes to an explanation of why acyclic-agreement forests appear, however, to be of little use to solve {\sc Minimum Hybridisation} for more than two trees. This result is a particular instance of a conjecture in~\cite[page 1626]{iersel16}.

For the purpose of the upcoming definitions, we regard the root of a binary phylogenetic $X$-tree $\cT$ as a vertex labelled $\rho$ at the end of a pendant edge adjoined to the original root. Furthermore, we view $\rho$ as an element of the leaf set of $\cT$; thus $\cL(\cT)=X\cup\{\rho\}$. Let $\cT$ and $\cT'$ be two binary phylogenetic $X$-trees.  An {\it agreement forest} $\cF=\{\cL_\rho,\cL_1,\cL_2,\ldots,\cL_k\}$ for $\cT$ and $\cT'$ is a partition of $X\cup \{\rho\}$ such that $\rho\in\cL_\rho$ and the following conditions are satisfied:
\begin{itemize}
\item[(i)] For all $i\in\{\rho,1,2,\ldots,k\}$, we have $\cT|\cL_i\cong\cT'|\cL_i$.
\item[(ii)] The trees in
$$\{\cT(\cL_i): i\in\{\rho,1,2,\ldots,k\}\}$$
and
$$\{\cT'(\cL_i): i\in\{\rho,1,2,\ldots,k\}\}$$
are vertex-disjoint subtrees of $\cT$ and $\cT'$, respectively.
\end{itemize}

\noindent Now, let $\cF=\{\cL_\rho,\cL_1,\cL_2,\ldots,\cL_k\}$ be an agreement forest for $\cT$ and $\cT'$. Let $G_{\cF}$ be the directed graph that has vertex set $\cF$ and an arc from $\cL_i$ to $\cL_{j}$ precisely if $i\neq j$ and

\begin{itemize}
\item[(iii)]  the root of $\cT(\cL_i)$ is an ancestor of the root of $\cT(\cL_{j})$ in $\cT$, or the root of $\cT'(\cL_i)$ is an ancestor of the root of $\cT'(\cL_{j})$ in $\cT'$. 
\end{itemize}
\noindent We call $\cF$ an {\it acyclic-agreement forest} for $\cT$ and $\cT'$ if $G_{\cF}$ has no directed cycle. Moreover, if $\cF$ contains the smallest number of elements over all acyclic-agreement forests for $\cT$ and $\cT'$, we say that $\cF$ is a {\it maximum acyclic-agreement forest} for $\cT$ and $\cT'$, in which case, we denote this number minus one by $m_a(\cT, \cT')$. 

Baroni et al.~\cite{baroni05} established the following characterisation for when a collection of binary phylogenetic $X$-trees contains exactly two trees.

\begin{theorem}
Let $\cP=\{\cT,\cT'\}$ be a collection of two binary phylogenetic $X$-trees. Then $h(\cP)=m_a(\cT,\cT')$.
\end{theorem}

Let $\cN$ be a phylogenetic network on $X$ with root $\rho$ that displays a set $\cP$ of binary phylogenetic $X$-trees. As above, we regard $\rho$ as a vertex at the end of a pendant edge adjoined to the original root. We obtain a forest from $\cN$ by deleting all reticulation edges, repeatedly contracting edges where one end-vertex has degree one and is not in $X\cup \{\rho\}$, deleting isolated vertices not in $X\cup \{\rho\}$ and, lastly, suppressing all vertices with in-degree one and out-degree one. Let $\{\cS_\rho,\cS_1,\cS_2,\ldots,\cS_k\}$ be the forest obtained from $\cN$ in this way. We say that $\cF=\{\cL(\cS_\rho),\cL(\cS_1),\cL(\cS_2),\ldots,\cL(\cS_k)\}$ is the forest {\it induced} by $\cN$. Moreover, $\cF$ is said to be {\it optimum} if $\cN$ is a tree-child network with $h_{\rm tc}(\cP)=h(\cN)$. For example, up to regarding $\rho$ as a new vertex that is adjoined to the original root of the two phylogenetic networks $\cN$ and $\cN'$ shown in Figure~\ref{fig:cps}, $$\cF=\{\{\rho,1,2,6\},\{3\},\{4\},\{5\}\} \textnormal{ and } \cF'=\{\{\rho,1,3,5,6\},\{2\},\{4\}\}$$ is the induced forest of $\cN$ and $\cN'$, respectively.
%\blue{Note that for $|\cP|=2$, the definition of an optimum forest for $\cP$ coincides with the definition of a maximum acyclic-agreement forest for $\cP$.}

In~\cite{iersel16}, the authors investigate {\sc Minimum Hybridisation} for three trees and conjecture that, given a set $\cP$ of three binary phylogenetic $X$-trees and the induced forest of a phylogenetic network $\cN$ that displays $\cP$ and $h(\cP)=h(\cN)$, it is NP-hard to determine $\cN$. For $|\cP|\geq 3$, we affirmatively answer their conjecture in the context of tree-child networks.
%and for when one wants to determine $h_{\rm tc}(\cP)$ without constructing $\cN$. 
More precisely, using cherry-picking sequences, we show that the following decision problem is NP-complete. \\

\noindent{\sc Scoring Optimum Forest} \\
\noindent {\bf Instance.} A non-negative integer $k$, a collection $\cP$ of binary phylogenetic $X$-trees, an optimum forest $\cF$ induced by a tree-child network $\cN$ on $X$ that displays $\cP$. \\
%A non-negative integer $k$, a collection $\cP$ of binary phylogenetic $X$-trees, an acyclic-agreement forest $\cF$ induced by a tree-child network $\cN$ on $X$ that displays $\cP$ and $h_{\rm tc}(\cP)=h(\cN)$. \\
\noindent {\bf Question.} Is $h_{\rm tc}(\cP)\le k$?\\

If $|\cP|=2$, then, by Observation~\ref{ob:tc}, $h_{\rm tc}(\cP)=h(\cP)$ and  $\cF$ is a maximum acyclic-agreement forest with $h_{\rm tc}(\cP)=|\cF|-1$. Hence, {\sc Scoring Optimum Forest} is polynomial time when $|\cP|=2$. However, the general problem is NP-complete.

%Let $\cP$ be a collection of binary phylogenetic $X$-trees, and let $$h_{\rm tc}(\cP)=\min\{h(\cN): \mbox{$\cN$ is a tree-child network on $X$ displaying $\cP$}\}.$$

\begin{theorem}
The problem {\sc Scoring Optimum Forest} is {\rm NP}-\-complete.
\label{maaf}
\end{theorem}

The remainder of this section consists of the proof of Theorem~\ref{maaf}. To establish the result, we use a reduction from a particular instance of the NP-complete problem {\sc Shortest Common Supersequence}. Let $\Sigma$ be a finite alphabet, and let $W$ be a finite subset of words in $\Sigma^*$. A word $z\in \Sigma^*$ is a {\it common supersequence of $W$} if each word in $W$ is a subsequence of $z$.\\

\noindent {\sc Shortest Common Supersequence (SCS)} \\
\noindent {\bf Instance.} A non-negative integer $k$, a finite alphabet $\Sigma$, and a finite subset $W$ of words in $\Sigma^*$. \\
\noindent {\bf Question.} Is there a supersequence of the words in $W$ with at most $k$ letters?\\

Timkovskii~\cite[Theorem 2]{tim89} established the next theorem. The {\it orbit} of a letter in $\Sigma$ is the set of
its occurrences in the words in $W$. Note that if a word in $W$ uses a letter, $b$ say, twice, then that word contributes two occurrences to the orbit of $b$.

\begin{theorem}
The decision problem {\sc SCS} is {\rm NP}-complete even if each word in $W$ has $3$ letters and the size of all orbits is $2$.
\label{scs1}
\end{theorem}

A consequence of Theorem~\ref{scs1} is the next corollary.

\begin{corollary}
The decision problem {\sc SCS} is {\rm NP}-complete even if each word in $W$ has $3$ letters, the size of all orbits is at most $2$, and no word  in $W$ contains a letter twice.
\label{scs2}
\end{corollary}

\begin{proof}
Let $k$, $\Sigma$, and $W$ be an instance of {\sc SCS}, where each word $W$ has $3$~letters and all orbits have size $2$. Let $Y$ be the subset of $W$ that consists of those words in $W$ in which no letter occurs twice. Observe that if $b\in \Sigma$ and $b$ occurs twice in a word in $W$, then no word in $Y$ contains $b$. Furthermore, with regards to $Y$, each word has $3$ letters, the size of all orbits is at most $2$, and no word contains a letter twice. Let $t$ denote the number of distinct letters that occur in two distinct words in $W-Y$.
%Let $t$ denote the number of distinct letters occurring exactly once in a word in $W-Y$ but which \blue{do} not occur in a word in $Y$. 
Note that the construction of $Y$ and the computation of $t$ can both be done in time polynomial in $|W|$. The corollary will follow from Theorem~\ref{scs1} by showing that {\sc SCS} with parameters $\Sigma$ and $W$ has a supersequence of length at most
$$2|W-Y|+t+k$$
if and only if {\sc SCS} with parameters $\Sigma$ and $Y$ has a supersequence of length at most $k$.

Suppose that {\sc SCS} with parameters $\Sigma$ and $Y$ has a supersequence $z$ of length at most $k$. Now iteratively extend $z$ to a sequence $z'$ as follows. Let $w\in W-Y$. Then $w$ contains two occurrences of a letter, $b$ say, in $\Sigma$. Let $d$ denote the third letter in $w$. Note that $b$ occurs in no other word in $\Sigma$ and $d$ occurs in exactly one other word in $\Sigma$. First assume that $d$ occurs in a word in $Y$. Depending on whether $d$ is the first, second, or third letter in $w$, extend $z$ by adding $bb$ to the end of $z$, adding $b$ at the beginning and $b$ at the end of $z$, or adding $bb$ at the beginning of $z$, respectively. The resulting sequence is a supersequence for $Y\cup \{w\}$. Second assume that $d$ does not occur as a word in $Y$. Then $d$ occurs in a word $w'$ in $W-Y$. Let $c$ denote the letter occurring twice in $w'$. Extend $z$ by adding $w$ to the beginning of $z$ and then, to the resulting sequence, add two occurrences of $c$ after $w$, add one occurrence of $c$ before $w$ and one occurrence of $c$ after $w$, or two occurrences of $c$ before $w$ depending on whether $d$ is the first, second, or third letter of $w'$, respectively. The resulting sequence is a supersequence for $Y\cup \{w, w'\}$.
%Extend $z$ by adding two occurrences of $b$, two occurrences of $c$, and one occurrence of $d$ to the beginning of $z$ so that the resulting sequence is a supersequence for $Y\cup \{w, w'\}$. It is easily checked that one can add these letters to the beginning of $z$ in this way. 
Taking the resulting sequence and repeating this process for each remaining word in $W-(Y\cup \{w\})$ or $W-(Y\cup \{w, w'\})$, respectively, we eventually obtain a supersequence $z'$ for $W$. Moreover, $z'$ has length
$$2|W-Y|+t+k.$$

For the converse, suppose that there is a supersequence $z$ of $W$ of length
$$2|W-Y|+t+k.$$
Let $z'$ be the sequence obtained from $z$ by deleting each occurrence of a letter that occurs twice in a word in $W$ and deleting exactly one occurrence of a letter that occurs in two distinct words in $W-Y$ and, hence, does not occur in a word in $Y$. It is easily checked that $z'$ is a supersequence of $Y$. Furthermore, since there are $2|W-Y|$ deletions of the first type and $t$ deletions of the second type, it follows that $z'$ has length $k$. This completes the proof of the corollary.\qed
\end{proof}

The decision problem described in the statement of Corollary~\ref{scs2} is the one we will use for the reduction in proving Theorem~\ref{maaf}.  Let $k$, $\Sigma$, and $W$ be an instance of {\sc SCS} such that each word in $W$ has $3$ letters, the size of all orbits is at most $2$, and no word in $W$ contains a letter twice. Without loss of generality, we may assume that, for each $\ell\in \Sigma$, there is a word in $W$ containing $\ell$, and that $|W|\ge 3$, so no letter is contained in each word. Let
$$W=\{w_1, w_2, \ldots, w_q\}$$
and, for each $i\in \{1, 2, \ldots, q\}$, let $w_i=w_{i1}\, w_{i2}\, w_{i3}$. Also, let
$$o(\Sigma)=\ell_1, \ell_2, \ldots, \ell_{|\Sigma|}$$
denote a fixed ordering of the letters in $\Sigma$. For each $w_i$ in $W$, we denote the sequence obtained from $o(\Sigma)$ by removing each of the three letters in $w_i$ by $o(\Sigma)-w_i$.

We now construct an instance of {\sc Scoring Optimum Forest}. A {\it rooted caterpillar} is a binary phylogenetic tree $\cT$ whose leaf set can be ordered, say $x_1, x_2, \ldots, x_n$, so that $\{x_1, x_2\}$ is a cherry and if $p_i$ denotes the parent of $x_i$, then, for all $i\in \{3, 4, \ldots, n\}$, we have $(p_i, p_{i-1})$ as an edge in $\cT$. Here, we denote the rooted caterpillar by $(x_1, x_2, \ldots, x_n)$.

Now, for each $i\in \{1, 2, \ldots, q\}$, let $\cT_i$ denote the rooted caterpillar
$$\cT_i=(\alpha, w_{i1}, w_{i2}, w_{i3}, \beta_1, \beta_2, \ldots, \beta_{q|\Sigma|}, o(\Sigma)-w_i),$$
and let $\cP=\{\cT_1, \cT_2, \ldots, \cT_q\}$. Note that each of the trees in $\cP$ has leaf set
$$X=\{\alpha, \beta_1, \beta_2, \ldots, \beta_{q|\Sigma|}\}\cup \Sigma$$
and $\cP$ can be constructed in time polynomial in the size of $\Sigma$ and $W$.

We next establish a lemma that reveals a relationship between the weight of a cherry-picking sequence for $\cP$ and the length of a supersequence for $W$. Let $$\sigma=(x_1, y_1), (x_2, y_2), \ldots, (x_s, y_s), (x_{s+1}, -)$$ be a cherry-picking sequence for a set $\cP$ of binary phylogenetic $X$-trees. For each $i\in\{1,2,\ldots,s\}$, we say that $(x_i,y_i)$ {\it corresponds} to $n$ trees in $\cP$ if $\{x_i,y_i\}$ is a cherry in exactly $n$  trees obtained from $\cP$ by picking $x_1,x_2,\ldots,x_{i-1}$, where $1\leq n\leq |\cP|$.

\begin{lemma}
Let $k\le 2|\Sigma|$ be a positive integer. Then there is a cherry-picking sequence of $\cP$ of weight $k$ if and only if there is a supersequence of $W$ of length $k$.
\label{super}
\end{lemma}

\begin{proof}
First suppose there is a common supersequence $z$ of $W$ of length $k$. Let
$$z=m_1\, m_2\, \cdots\, m_k,$$
and let $\sigma$ denote the sequence
$$(m_1, \alpha), \ldots, (m_k, \alpha), (\beta_1, \alpha), \ldots, (\beta_{q|\Sigma|}, \alpha), (\ell_1, \alpha), \ldots, (\ell_{|\Sigma|}, \alpha), (\alpha, -).$$
Since $z$ is a supersequence of $W$, it is easily seen that $\sigma$ is a cherry-picking sequence of $\cP$. Moreover,
$$w(\sigma)=(k+|X|)-|X|=k.$$

Now suppose that there is a cherry-picking sequence $$\sigma=(x_1, y_1), (x_2, y_2), \ldots, (x_s, y_s), (x_{s+1}, -)$$ of $\cP$ of weight $k$. Without loss of generality, we may assume that each ordered pair in $\sigma$ is essential. We first show that there is a positive integer $i'$ such that $\cP_{i'}$ is obtained from $\cP$ by picking $x_1,x_2,\ldots, x_{i'}$ and each tree in $\cP_{i'}$
%after $i'$ iterations of {\sc Picking Cherries} applied to $\cP$ and $\sigma$, each of the trees in $\cP_{i'}$ 
has a cherry consisting of two elements in $\{\alpha, \beta_1, \beta_2, \ldots, \beta_{q|\Sigma|}\}$. 
If not, then there is a word $w_i$ in $W$ such that either $(w_{ij}, \beta_{q|\Sigma|})$ or $(\beta_{q|\Sigma|}, w_{ij})$ is an ordered pair in $\sigma$, where $j\in \{1, 2, 3\}$. Now, by considering a word not containing $w_{ij}$ and its associated tree in $\cP$, it is easily seen that $w(\sigma)\ge q|\Sigma|-1$ as each of the elements in $\{\beta_1, \beta_2, \ldots, \beta_{q|\Sigma|-1}\}$ appear at least twice as the first element of an ordered pair in $\sigma$. But then
$$q|\Sigma|-1 > 2|\Sigma|$$
as $q\ge 3$ and $|\Sigma|\ge 3$, contradicting the assumption $k\le 2|\Sigma|$.

Consider the first $i'$ ordered pairs in $\sigma$. For each tree $\cT_i$ in $\cP$, there are exactly three ordered pairs whose first and second elements are in $$S=\{\alpha, w_{i1}, w_{i2}, w_{i3}\}$$ and picking $x_1,x_2,\ldots,x_{i'}$ from $\cT_i$ picks three elements of $S$.
%and for which {\sc Picking Cherries} picks three of these elements in $\cT_i$. 
Since the size of all orbits is at most $2$, such an ordered pair corresponds to at most two trees in $\cP$. We next construct a sequence $\sigma'$ of ordered pairs obtained from $\sigma$. We start by modifying the first $i'$ ordered pairs of $\sigma$ as follows:\\
\begin{enumerate}[(a)]
\item Amongst the first $i'$ ordered pairs, replace each ordered pair of the form $(\alpha, \ell)$ with $(\ell, \alpha)$, where $\ell\in \Sigma$.

\item With the sequence obtained after (a) is completed, sequentially move along the sequence to the $i'$-th ordered pair replacing each ordered pair of the form $(\ell, \ell')$, where $\ell, \ell'\in \Sigma$ in one of the following ways:
\begin{enumerate}[(i)]
\item If $(\ell, \ell')$ corresponds to exactly one tree in $\cP$, then replace it with $(\ell, \alpha)$ or $(\ell', \alpha)$ depending on whether $(\ell', \alpha)$ or $(\ell, \alpha)$, respectively, is an earlier ordered pair.

\item If $(\ell, \ell')$ corresponds to two trees, $\cT_i$ and $\cT_j$ say, in $\cP$ and the order of the letters $\ell$ and $\ell'$ is the same in $w_i$ and $w_j$, then replace it with $(\ell, \alpha)$ or $(\ell', \alpha)$ depending on whether $(\ell', \alpha)$ or $(\ell, \alpha)$, respectively, is an earlier ordered pair.

\item If $(\ell, \ell')$ corresponds to two trees, $\cT_i$ and $\cT_j$ say, in $\cP$ and the order of the letters $\ell$ and $\ell'$ in $w_i$ is not the same as that in $w_j$, then replace it with $(\ell, \alpha)$ if $(\ell, \alpha)$ occurs as an ordered pair before $(\ell', \alpha)$ earlier in the sequence; otherwise, $(\ell', \alpha)$ occurs as an ordered pair before $(\ell, \alpha)$ earlier in the sequence and so replace it with $(\ell', \alpha)$.\\
\end{enumerate}

%\item With the sequence obtained after (b) is completed, amongst the first $i'$ ordered pairs, replace each ordered pair of the form $(\ell, \beta_i)$ or $(\beta_i, \ell)$, where $\ell\in \Sigma$ and $i\in \{1, 2, \ldots, q|\Sigma|\}$, with $(\alpha, \beta_i)$ and $(\beta_i, \alpha)$, respectively.
\end{enumerate}
With this modification of $\sigma$ after (b) is completed, let $\sigma'_1$ denote the subsequence of the first $i'$ ordered pairs whose coordinates are in $\Sigma\cup \{\alpha\}$, and let $\sigma'_2$ denote the subsequence $\sigma-\sigma'_1$. Let $\sigma'$ denote the concatenation of $\sigma'_1$ and $\sigma'_2$.

Now, consider each tree $\cT_i$ in $\cP$ together with its corresponding ordered pairs in $\sigma$ and the associated ones in $\sigma'$. Let $\cP_{|\sigma'_1|}$ be the set of trees obtained from $\cP$ by picking $x_1',x_2',\ldots,x_{|\sigma'_1|}'$, where $x_j'$ is the first coordinate of the $j$-th ordered pair in $\sigma'$ for each $j\in\{1,2,\ldots,|\sigma'_1|\}$. A routine check shows that this picking sets the tree corresponding to $\cT_i$ in $\cP_{|\sigma'_1|}$ to be the rooted caterpillar
%By considering each tree $\cT_i$ in $\cP$ together with its corresponding ordered pairs in $\sigma$ and the associated ones in $\sigma'$, a routine check shows that the first $|\sigma'_1|$ iterations of {\sc Picking Cherries} applied to $\sigma'$ and $\cP$ sets the tree corresponding to $\cT_i$ in $\cP_{|\sigma'_1|}$ to be the rooted caterpillar
$$(\alpha, \beta_1, \beta_2, \ldots, \beta_{q|\Sigma|}, o(\Sigma)-w_i).$$
In particular, $(w_{i1}, \alpha), (w_{i2}, \alpha), (w_{i3}, \alpha)$ is subsequence of $\sigma'_1$.

%By considering each tree $\cT_i$ in $\cP$ together with its corresponding ordered pairs in $\sigma$ and the associated ones in $\sigma'$, a routine check shows that $\sigma'$ is a cherry-picking sequence of $\cP$. \blue{In particular, $(w_{i1}, \alpha), (w_{i2}, \alpha), (w_{i3}, \alpha)$ is a subsequence of $\sigma'$ for each $\cT_i$.}
%%In particular, the associated ordered pairs in $\sigma'$ induces the subsequence $(w_{i1}, \alpha), (w_{i2}, \alpha), (w_{i3}, \alpha)$.

We next extend $\sigma'_1$ to a cherry-picking sequence for $\cP$ of weight at most $k$. Consider $\sigma$ and $\sigma'$. If $\sigma_1$ denotes the subsequence of ordered pairs in $\sigma$ corresponding to $\sigma'_1$, then $|\sigma_1|=|\sigma'_1|$ and
$$|\sigma|-|\sigma_1|\ge q|\Sigma|+|\Sigma|+1$$
as each of the elements in $\{\beta_1, \beta_2, \ldots, \beta_{q|\Sigma|}\}\cup \Sigma$ as well as at least one element in $\{\beta_1, \beta_2, \ldots, \beta_{q|\Sigma|}\}\cup \Sigma\cup \{\alpha\}$ appears as the first coordinate of an ordered pair in $\sigma-\sigma_1$. Here, an element in $\{\beta_1, \beta_2, \ldots, \beta_{q|\Sigma|}\}\cup \Sigma$ may be counted twice as it appears as the first coordinate of two ordered pairs in $\sigma-\sigma_1$. It follows that the sequence of ordered pairs that is the concatenation of $\sigma'_1$ and
$$(\beta_1, \alpha), (\beta_2, \alpha), \ldots, (\beta_{q|\Sigma|}, \alpha), (\ell_1, \alpha), (\ell_2, \alpha), \ldots, (\ell_{|\Sigma|}, \alpha), (\alpha, -)$$
is a cherry-picking sequence of $\cP$ whose weight is at most $k$.

Let $\sigma'_1$ be the sequence
$$(m_1, \alpha), (m_2, \alpha), \ldots, (m_{k'}, \alpha).$$
%obtained from $\sigma'$ by removing all ordered pairs after the $i'$-th ordered pair, and then removing any ordered pair that is not of the form $(\ell, \alpha)$, where $\ell\in \Sigma$. It is easily seen that, as $\sigma'$ is a cherry-picking sequence of $\cP$,
Since $(w_{i1}, \alpha), (w_{i2}, \alpha), (w_{i3}, \alpha)$ is a subsequence of $\sigma'_1$ for each tree $\cT_i$,
$$m_1\, m_2\, \cdots\, m_{k'}$$
is a common supersequence of $W$. Moreover, as $w(\sigma)=k$, we have $k'\le k$. It follows that there is a supersequence of $W$ of length $k$.\qed
\end{proof}

To complete the proof of Theorem~\ref{maaf}, let
$$\cF=\big\{\{\rho, \alpha, \beta_1, \beta_2, \ldots, \beta_{q|\Sigma|}\}\big\}\cup \big\{\{\ell\}: \ell\in \Sigma\big\}$$
be a partition of $X\cup \{\rho\}$. We next show that $\cF$ is an optimum forest induced by a tree-child network on $X$ with root $\rho$ and that displays $\cP$. 
%It is easily checked that $\cF$ is acyclic. 
Let $z=z_1z_2\cdots z_k$ be a common supersequence of $W$ of minimum length, and suppose this length is $k$. Since all orbits have size at most $2$ and $z$ is of minimum length, each letter in $\Sigma$ appears at most twice in $z$, and so $k\le 2|\Sigma|$. Let $\cT$ be the `multi-labelled' rooted caterpillar
$$\cT=(\alpha,\, z_1,z_2,\ldots,z_k,\, \beta_1,\, \beta_2,\, \ldots,\, \beta_{q|\Sigma|},\, o(\Sigma))$$
and let $\cN$ be the tree-child network with root $\rho$ obtained from $\cT$ as follows. For each $\ell\in \Sigma$, identify the leaves labelled $\ell$ and adjoin a new pendant edge to the identified vertex with the leaf-end labelled $\ell$. Since $z$ is a common supersequence of $W$, it is easily checked that $\cN$ displays $\cP$. Furthermore, $h(\cN)=k$. By Theorem~\ref{picking1} and Lemma~\ref{super}, $h_{\rm tc}(\cP)=k$, and so, as $\cF$ is induced by $\cN$, it follows that $\cF$ is an optimum forest for $\cP$. 

Now, given an arbitrary phylogenetic network, it can be verified in polynomial time whether it is tree-child, it displays $\cP$~\cite{simpson}, its hybridisation number is at most $k$, and it induces $\cF$. Hence, {\sc Scoring Optimum Forest} is in NP.

Theorem~\ref{maaf} now follows by combining Corollary~\ref{scs2} with Theorem~\ref{picking1} and Lemma~\ref{super}.

\section{Concluding remarks}\label{sec:conclu}

In this paper, we have generalised the concept of cherry-picking sequences as introduced in~\cite{humphries13} and shown how this generalisation can be used to characterise the minimum number of reticulation events that is needed to explain any set of phylogenetic $X$-trees in the space of tree-child networks as well as in the space of all phylogenetic networks. To see that these two minima can be different for a fixed set of phylogenetic trees, consider the set $\cP$ of trees presented in Figure~\ref{fig:leaf-added}. It was shown in~\cite{iersel16,kelk12COM} that $h(\cP)=3$  and that there are six phylogenetic networks each of which displays $\cP$ and has a hybridisation number of three. However, none of these six phylogenetic networks is tree-child. Moreover, using cherry-picking sequences a  straightforward check shows that  $h_{\rm tc}(\cP)=4$. Furthermore, we have shown that {\sc Scoring Optimum Forest} is NP-complete. Hence, given an optimum forest, it is computationally hard to compute $h_{\rm tc}(\cP)$ for when $\cP$ is a set of binary phylogenetic $X$-trees, where $|\cP|\ge 3$. This contrasts with the two-tree case for which  {\sc Scoring Optimum Forest} is polynomial-time solvable and  further hints  at that agreement forests are of limited use beyond the two-tree case.

Of course, restricting to collections of binary phylogenetic trees, one could generalise the definition of an acyclic-agreement forest for two binary phylogenetic trees to more than two trees in the most obvious way. That is, one requires Conditions (i), (ii), and (iii) in the definition of an acyclic-agreement forest to hold for each tree in an arbitrarily large collection of binary phylogenetic $X$-trees. With this generalisation in mind and observing that the number of components in a forest that is induced by a tree-child network is equal to its number of reticulations plus one, one might conjecture that, given a set $\cP$ of binary phylogenetic $X$-trees, the number of components in a maximum acyclic-agreement forest for $\cP$ is the same as the minimum number of components in an optimum forest for $\cP$. To see that this is not true, we refer back to Figure~\ref{fig:cps}. Let $\cF$ be the forest induced by $\cN$, and let $\cF'$ be the forest induced by $\cN'$. Since $|\cF'|=3$, a maximum acyclic-agreement forest for $\cP$ has at most three elements. Moreover, since $h(\cN)$=3 and $\cN$ is tree-child, we have $h_{\rm tc}(\cP)\leq 3$. Indeed, it can be checked that $h_{\rm tc}(\cP)= 3$. Moreover, there is no tree-child network that displays $\cP$ and induces an optimum forest that is also a maximum acyclic-agreement forest for $\cP$. Consequently, an approach that exploits maximum acyclic-agreement forests for a set $\cP$ of binary phylogenetic trees to compute $h_{\rm tc}(\cP)$, such as computing a maximum acyclic-agreement forest $\cF$ for $\cP$ and, subsequently, scoring $\cF$ in a way that reflects the number of edges that are directed into each reticulation vertex in a network that induces $\cF$, is unlikely to give the desired result. 

Lastly, from a computational viewpoint, the introduction of acyclic-agree\-ment forests~\cite{baroni05} has triggered significant progress towards the development of ever faster algorithms to solve {\sc Minimum Hybridisation} for when the input contains exactly two phylogenetic trees (e.g. see~\cite{albrecht12,bordewich07b,chen13,collins11,piovesan12,wu10}). We look forward to seeing a similar development now for solving {\sc Minimum Hybridisation} for  arbitrarily many phylogenetic trees by using cherry-picking sequences. In turn, this is likely to be of benefit to biologists who often wish to infer evolutionary histories that are not entirely tree-like and for data sets that usually consists of more than two phylogenetic trees.\\

%
% 
%\section{Questions for us not for the windows of various shops}
%
%\noindent Can we bound $h_{tc}(\cP)-h(\cP)$?
%
%\noindent Can we bound the number of additional leaves in a leaf-added cherry-picking sequence?
%
%\noindent Is there a canonical construction to make $N$ tree child without adding new leaves?
%
%\noindent Gives new approach for the two tree problem.
%
%
%\section{Overflow - Agreement forests and tree-child networks}
%October 2016 examples showing that forests do not work for more than two trees. Even if we restrict to tree-child networks.
%

%%%%%%%%%%%%%%%%%%%%%%%%%%%%%%%%%%%%%%%%%%

\noindent {\bf Acknowledgements.} We thank the New Zealand Marsden Fund for their financial support. \\

\noindent {\bf References.}

\end{document}